\documentclass[11pt,paper=a4,DIV=12,abstract=true]{scrartcl}

\usepackage[T1]{fontenc}
\usepackage[utf8]{inputenc}
\usepackage{lmodern}
\usepackage{microtype}

\usepackage{amsmath,amssymb,amsthm,mathtools}
\usepackage{mathrsfs}
\usepackage{bm}

\usepackage{graphicx}
\usepackage{booktabs}
\usepackage{threeparttable}

\usepackage{aliascnt}
\usepackage[dvipsnames]{xcolor}
\definecolor{linkblue}{RGB}{0,45,90}

\usepackage{algorithm}
\usepackage{algorithmic}

\usepackage[round,authoryear]{natbib}
\usepackage[unicode=true,colorlinks=true,linkcolor=linkblue,citecolor=linkblue,urlcolor=linkblue]{hyperref}
\usepackage[nameinlink,noabbrev,capitalize]{cleveref}
\setkomafont{disposition}{\normalfont\bfseries}
\setkomafont{title}{\normalfont\bfseries}
\setkomafont{author}{\normalfont}
\setkomafont{date}{\normalfont}
\setkomafont{caption}{\small}
\setkomafont{captionlabel}{\small\bfseries}

\numberwithin{equation}{section}

\theoremstyle{plain}
\newtheorem{theorem}{Theorem}[section]

\newaliascnt{proposition}{theorem}
\newtheorem{proposition}[proposition]{Proposition}
\aliascntresetthe{proposition}

\newaliascnt{lemma}{theorem}
\newtheorem{lemma}[lemma]{Lemma}
\aliascntresetthe{lemma}

\newaliascnt{corollary}{theorem}

\aliascntresetthe{corollary}

\theoremstyle{definition}

\newaliascnt{assumption}{theorem}

\aliascntresetthe{assumption}

\newaliascnt{definition}{theorem}
\newtheorem{definition}[definition]{Definition}
\aliascntresetthe{definition}

\newaliascnt{example}{theorem}
\newtheorem{example}[example]{Example}
\aliascntresetthe{example}

\theoremstyle{remark}

\newaliascnt{remark}{theorem}
\newtheorem{remark}[remark]{Remark}
\aliascntresetthe{remark}

\crefname{theorem}{Theorem}{Theorems}
\Crefname{theorem}{Theorem}{Theorems}
\crefname{proposition}{Proposition}{Propositions}
\Crefname{proposition}{Proposition}{Propositions}
\crefname{lemma}{Lemma}{Lemmas}
\Crefname{lemma}{Lemma}{Lemmas}
\crefname{corollary}{Corollary}{Corollaries}
\Crefname{corollary}{Corollary}{Corollaries}
\crefname{assumption}{Assumption}{Assumptions}
\Crefname{assumption}{Assumption}{Assumptions}
\crefname{definition}{Definition}{Definitions}
\Crefname{definition}{Definition}{Definitions}
\crefname{example}{Example}{Examples}
\Crefname{example}{Example}{Examples}
\crefname{remark}{Remark}{Remarks}
\Crefname{remark}{Remark}{Remarks}

\DeclareMathOperator{\Var}{Var}

\newcommand{\R}{\mathbb{R}}
\newcommand{\E}{\mathbb{E}}

\DeclarePairedDelimiter{\abs}{\lvert}{\rvert}
\DeclarePairedDelimiter{\norm}{\lVert}{\rVert}

\newcommand{\keywords}[1]{\par\medskip\noindent\textbf{Keywords.} #1}
\newcommand{\msc}[1]{\par\smallskip\noindent\textbf{MSC 2020.} #1}
\newcommand{\jel}[1]{\par\smallskip\noindent\textbf{JEL classification.} #1}

\hypersetup{
  pdftitle={On the approximation of posterior laws in compound loss models by conditional Wasserstein GANs},
  pdfauthor={Aleksandar Arandjelović and Pavel V. Shevchenko and George Tzougas},
    pdfkeywords={amortized Bayesian inference; conditional Wasserstein GANs; compound loss models; simulation-based calibration; Bayesian credibility; natural catastrophe risk}
}

\title{{\Large On the Approximation of Posterior Laws in Compound Loss Models by Conditional Wasserstein GANs}}

\author{%
Aleksandar Arandjelović\thanks{Department of Mathematics, ETH Zürich, Zürich, Switzerland, and Institute for Statistics and Mathematics, Vienna University of Economics and Business, Vienna, Austria. Email: \href{mailto:aleksandar.arandjelovic@math.ethz.ch}{aleksandar.arandjelovic@math.ethz.ch}. Corresponding author.}
\and
Pavel V. Shevchenko\thanks{Department of Actuarial Studies and Business Analytics, Macquarie University, Sydney, Australia. Email: \href{mailto:pavel.shevchenko@mq.edu.au}{pavel.shevchenko@mq.edu.au}.}
\and
George Tzougas\thanks{School of Mathematical and Computer Sciences \& the Maxwell Institute for Mathematical Sciences, Department of Actuarial Mathematics and Statistics \& Duncan Laboratory for Insurance Data Science, Heriot-Watt University, Edinburgh, UK. Email: \href{mailto:george.tzougas@hw.ac.uk}{george.tzougas@hw.ac.uk}.}
}

\date{\today}

\begin{document}

\maketitle

\begin{abstract}
Bayesian inference in compound loss models must often be repeated across policies, market scenarios, and prior specifications.
Outside conjugate cases, this may require repeated numerical integration or Markov chain Monte Carlo (MCMC).
We formulate this problem as amortized posterior approximation and construct a conditional Wasserstein generative adversarial network conditioned on sufficient statistics, prior mean and coefficient of variation, and mixture weights of prior families.
Notably, a single shared generator is able to approximate the posterior laws of both the Poisson intensity and the Pareto shape parameter under mixtures of Gamma, inverse-Gaussian, and lognormal priors.
We assess the approximation by simulation-based calibration and by comparisons with analytical posteriors, deterministic quadrature, and extensive MCMC simulations.
In an application to data on extreme natural catastrophe losses, we produce rolling one-year posterior predictive distributions, and examine the effects of heavy-tailed severity and prior-family uncertainty on aggregate tail risk.
\end{abstract}

\keywords{Amortized Bayesian inference; conditional Wasserstein GANs; compound loss models; simulation-based calibration; Bayesian credibility; natural catastrophe risk}
\msc{Primary 62P05; Secondary 62F15, 68T07}
\jel{C11; C15; C45; G22}

\clearpage

\section{Introduction}\label{sec:introduction}
Bayesian posteriors formalize a question that has been at the heart of insurance mathematics for more than a century: how to learn about heterogeneous risks from limited, incomplete, or noisy observations?
This classical actuarial problem dates back at least to Whitney's theory of experience rating \citep{whitney1918theory}, through Bailey's generalized credibility theory \citep{bailey1945generalized} and Longley-Cook's systematic introduction to credibility theory \citep{longleycook1962introduction}, to the explicitly Bayesian interpretation of \citet{mayerson1964bayesian}.
Typically, insurers observe claims, exposures, covariates, portfolio history, and market information, while the underlying risk characteristics are only partially observed.
The posterior law updates prior actuarial knowledge in light of these observations, while the posterior predictive law propagates the resulting uncertainty to future losses, whose quantiles, tail probabilities, and other risk measures may be of interest.
Classical credibility models can be viewed as a tractable instance of this general learning problem.
In particular, the B\"uhlmann and B\"uhlmann--Straub models yield the best linear predictors of the conditional risk mean under squared-error loss \citep{buhlmann1967experience,buhlmann1970glaubwurdigkeit,buhlmann2005course}.
For conditionally independent and identically distributed (i.i.d.) observations from a simple exponential family equipped with its natural conjugate prior, \citet{jewell1974acredible} further showed that the resulting credibility premium coincides with the Bayes posterior expectation of the conditional risk mean.

Examples of conjugate actuarial models include Poisson--Gamma, Binomial--Beta, and Normal--Normal under known variance.
There, the posterior remains in the same parametric family as the prior, and the posterior mean takes the familiar credibility form: prior information and observed experience are combined through a credibility factor.
\citet{bailey1950credibility} and \citet{mayerson1964bayesian} showed that the Bayesian posterior mean takes the exact linear credibility form for several special prior--likelihood combinations, while \citet{jewell1974acredible} placed these cases within the broader structure of simple exponential families with natural conjugate priors.
However, this analytical convenience comes at the cost of modelling flexibility, as the likelihood and prior must be chosen from compatible parametric families to retain analytically tractable Bayesian updating.
For this reason, conjugate credibility models are best viewed as highly tractable benchmark cases for Bayesian actuarial learning, against which modern hierarchical, simulation-based, and machine-learning posterior approximations can be assessed.

The trade-off between conjugacy and modelling flexibility becomes particularly pronounced in actuarial models with richer data structures and more complex dependence.
In Bayesian reserving, \citet{verrall1990bayes} embedded the chain-ladder model in a hierarchical Bayesian framework, deriving empirical Bayes estimators that admit a credibility interpretation.
With hierarchical portfolio structures, covariates, censoring and truncation, deductibles and policy limits, reporting delays, dependence, or heavy-tailed severities, conjugacy is generally lost, and posterior predictive distributions rarely admit closed-form expressions.
The resulting loss of analytic tractability made computational methods increasingly important: \citet{makov1996bayesian} reviewed how advances in Bayesian computational methodology opened the way to a fully Bayesian treatment of certain actuarial problems, while \citet{scollnik1996mcmc,scollnik2001actuarial} presented and illustrated MCMC methods for actuarial applications, including their implementation through BUGS (Bayesian inference using Gibbs sampling).
The increased modelling flexibility of non-conjugate specifications, however, requires numerical posterior inference together with convergence diagnostics, sensitivity analyses, and posterior predictive checks.
Non-conjugate Bayesian models thus extend credibility beyond analytically tractable benchmark cases, while making accurate and efficient posterior approximation an important computational problem.

Bayesian variational inference provides one methodological approach to this problem.
Variational inference recasts posterior approximation as an optimization problem over a tractable family of distributions, typically by maximizing the evidence lower bound or, equivalently, minimizing the reverse Kullback--Leibler divergence to the posterior.
A systematic treatment of variational methods for approximate inference in graphical models was given by \citet{jordan1999introduction}, while \citet{wainwright2008graphical} developed variational representations of likelihoods and marginal probabilities using the duality between the cumulant generating function and the entropy for exponential families, and \citet{blei2017variational} subsequently surveyed variational inference as an optimization-based approach to approximate Bayesian inference.
More recent actuarial applications include \citet{kim2022approximation}, who use variational Bayes to approximate Bayesian credibility premiums in a zero-inflated Poisson frequency model with unobserved policyholder heterogeneity, while \citet{avanzi2024machine} rely on variational inference for parameter estimation in a generalized linear mixed model neural network designed for high-cardinality categorical features.
However, variational posterior approximation is not innocuous.
A restrictive variational family may understate posterior dispersion and distort dependence, predictive quantiles, and tail probabilities.
Its adequacy should therefore be examined by posterior predictive and calibration diagnostics, sensitivity analysis and, where feasible, comparison with more accurate numerical benchmarks.

Actuarial applications rarely concern a single posterior law in isolation.
Pricing, reserving, portfolio monitoring, and stress testing require posterior laws to be computed and updated across policies, portfolios, valuation dates, prior specifications, and market scenarios, often together with joint predictive laws for related risks or development cells.
Classical credibility theory already contains elements of this perspective.
\citet{norberg2004credibility} surveys credibility theory in terms of experience rating, Bayes and linear Bayes estimation, empirical Bayes methods, hierarchical models, and recursive computation, while \citet{hachemeister1975credibility} extended credibility to regression models by combining state-specific and countrywide trend information.
In claims reserving, \citet{england2002stochastic} emphasized the full predictive distribution of reserve outcomes, while \citet{dealba2002bayesian} and \citet{ntzoufras2002bayesian} developed Bayesian predictive models for outstanding claims and liabilities.
Although a single model fit may produce such a predictive law, the analysis must be updated as new development data arrive and repeated across portfolios, valuation dates, prior specifications, and stress scenarios.
Beyond conjugate models, posterior computation and validation may become prohibitively expensive when repeated over many such conditioning inputs.
This motivates learning the posterior or posterior predictive law as a conditional map from claims experience, exposure information, covariates, and prior inputs to an approximate distribution.

Let $\vartheta$ be a latent model parameter taking values in a parameter space $\Theta$, let $\xi$ collect actuarial conditioning information, such as prior hyperparameters, covariates, or stress scenarios, and write $\bm{Y}_n=(Y_1,\ldots,Y_n)$ for a vector of $n$ recorded observations.
Let $S_n$ denote a sufficient statistic for $\vartheta$ conditional on $\xi$, and write $s_n$ for its realized value.
Setting $c=(n,s_n,\xi)$, the computational target is the law-valued map
\begin{equation*}
c\mapsto P_c\coloneqq\mathcal L(\vartheta\mid n,S_n=s_n,\xi),
\end{equation*}
or, after mixing a claims law with respect to $P_c$, the corresponding posterior predictive law of future losses.
An amortized posterior model replaces repeated numerical inference by a fitted conditional approximation $Q_{\theta,c}$ of $P_c$ that can be evaluated at new conditioning values $c$.
A related form of amortization appears in auto-encoding variational Bayes, where \citet{kingma2014autoencoding} fit an approximate inference, or recognition, model to intractable posteriors.
More directly, \citet{radev2020bayesflow} construct a globally amortized posterior estimator from simulated data; see also the broader account of simulation-based inference in \citet{cranmer2020frontier}.
Thus, a single fitted parameter $\theta^\ast$ is shared across conditioning values, replacing repeated pointwise posterior computation by one conditional approximation.

Conditional Wasserstein generative adversarial networks (cWGANs), a conditional version of the generative-adversarial framework of \citet{goodfellow2014generative}, are appealing in this setting because they provide a direct representation of $Q_{\theta,c}$ through a conditional generator $G_\theta$.
Given $c=(n,s_n,\xi)$ and independent noise $Z\sim\nu_Z$, the generator induces the law
\begin{equation*}
Q_{\theta,c}=\bigl(G_\theta(c,\cdot)\bigr)_\#\nu_Z,
\end{equation*}
as the pushforward of $\nu_Z$ under $G_\theta(c,\cdot)$, and $\theta$ is learned by minimizing a neural approximation to the $1$-Wasserstein distance between $P_c$ and $Q_{\theta,c}$, averaged over the conditioning values.
The chosen notion of distance follows the Wasserstein formulation of \citet{arjovsky2017wasserstein}.
For additional background on the Wasserstein distance and its connection to optimal transport, we refer to \citet{villani2009optimal}.
By approximating $P_c$ itself rather than only selected posterior functionals, the cWGAN adopts an inferential target that is natural in actuarial mathematics, where distributions of aggregate losses and measures of reserve uncertainty have long been objects of direct calculation.
For example, \citet{panjer1981recursive} derived a recursion for the evaluation of compound distributions, while \citet{heckman1983calculation} gave an algorithm for calculating cumulative probabilities and excess pure premiums, and \citet{mack1993distribution} derived a distribution-free standard error for chain-ladder reserve estimates.

In this paper, we study conditional Wasserstein GANs as amortized approximations of posterior laws in Bayesian credibility models for claim frequency and severity.
We take the latent frequency or severity parameter as the object of inference, and posterior predictive loss distributions are obtained by mixing the corresponding claims law with respect to the generated posterior distribution.
Sufficient statistics are used to reduce the dimension of the conditioning information.
The accuracy of the resulting approximation is assessed in both conjugate and non-conjugate models against analytical and numerical reference posteriors.
Finally, a single generator is trained over mixtures of Gamma, inverse-Gaussian, and lognormal priors and is applied to historical disaster data to assess the risk of extreme natural-catastrophe losses.

The main contributions of this work are threefold:
\begin{enumerate}
\item We formulate repeated Bayesian inference in compound loss models as the amortized approximation of a conditional posterior law and develop a conditional Wasserstein GAN that generates posterior draws from sufficient statistics, prior moments, and prior-family mixture weights.
\item We construct a single conditional generator for Poisson frequency and Pareto severity models under mixtures of Gamma, inverse-Gaussian, and lognormal priors. Its accuracy is assessed in both conjugate and non-conjugate cases, using analytical posteriors where available and deterministic quadrature and MCMC benchmarks otherwise.
\item We apply the resulting framework to extreme natural-catastrophe losses recorded in the Emergency Events Database (EM-DAT, \citet{delforge2025emdat}), where we infer the prior-family mixture weights for frequency and severity, produce rolling one-year posterior-predictive distributions of aggregate losses, and examine the effects of heavy-tailed severity and prior-family uncertainty on catastrophe tail risk.
\end{enumerate}

The remainder of the paper is organized as follows.
In \cref{sec:problem_formulation}, we formulate the Bayesian inference problem and introduce the frequency and severity models considered throughout.
In \cref{sec:cwgan}, we develop the conditional Wasserstein GAN framework for amortized posterior approximation.
In \cref{sec:simulation_study}, we assess the numerical accuracy of the method and apply the fitted generator to extreme natural-catastrophe losses recorded in EM-DAT.
\Cref{sec:conclusion} concludes.

\section{Problem formulation}\label{sec:problem_formulation}
We consider claims experience arising under a variety of actuarial risk models, including the B\"uhlmann--Straub model and hierarchical models for claim frequency and severity.
Let $\vartheta$ denote a latent parameter taking values in a parameter space $\Theta\subseteq\R^d$, let $\xi$ collect actuarial conditioning information, such as prior hyperparameters, covariates, or stress scenarios, and let $\pi_\xi$ denote the prior law of $\vartheta$.
For a risk unit with $n\in\mathbb N$ recorded observations in a measurable space $(\mathsf Y,\mathcal Y)$, write $\bm{Y}_n=(Y_1,\ldots,Y_n)$, and, for each $j=1,\ldots,n$, let $\vartheta\mapsto F_{\vartheta,\xi,j}$ be a probability kernel from $\Theta$ to $(\mathsf Y,\mathcal Y)$ giving the conditional law of $Y_j$.
Under conditional independence given $(\vartheta,\xi)$, the joint law of $(\vartheta,\bm Y_n)$ conditional on $\xi$ is
\begin{equation*}
\pi_\xi(\mathrm{d}\vartheta)\prod_{j=1}^n F_{\vartheta,\xi,j}(\mathrm{d}y_j),
\qquad (\vartheta, \bm{y}_n)\in\Theta\times\mathsf Y^n.
\end{equation*}
In the frequency model, the index $j$ labels an exposure period (e.g., a year), whereas in the severity model it labels an observed claim, so $n$ is the number of exposure periods in the former and the number of observed claim severities in the latter.

For the frequency model, we take $\Theta=(0,\infty)$ and $\vartheta=\Lambda$, where $\Lambda$ denotes the latent annual claim intensity with prior law $\pi_\xi$.
Under unit exposure, the annual counts are conditionally independent given $(\Lambda,\xi)$ and satisfy
\begin{equation*}
Y_j\mid(\Lambda=\lambda,\xi)\sim\mathrm{Poisson}(\lambda),
\qquad j=1,\ldots,n.
\end{equation*}
For observations $\bm y_n=(y_1,\ldots,y_n)\in\mathbb N_0^n$, the likelihood of $\lambda$ is
\begin{equation*}
L_n(\lambda;\bm y_n)=\prod_{j=1}^n\frac{e^{-\lambda}\lambda^{y_j}}{y_j!}=\frac{e^{-n\lambda}\lambda^{y_+}}{\prod_{j=1}^n y_j!},\qquad y_+=\sum_{j=1}^n y_j.
\end{equation*}
Since the factor $(\prod_{j=1}^n y_j!)^{-1}$ is independent of $\lambda$, one may take
$S_n=Y_+\coloneqq\sum_{j=1}^nY_j$ and $s_n=y_+$.
Moreover, $Y_+\mid(\Lambda=\lambda,\xi)\sim\mathrm{Poisson}(n\lambda)$.
This simplification is due to properties of the Poisson distribution and conditional independence.
The cWGAN formulation developed below is not restricted to this setting, although richer dependence structures generally require a different conditioning summary.
We leave such extensions, for example through copula dependence or additional latent effects, to future work.

For the severity model, we fix a threshold $x_{\min}>0$, take $\Theta=(0,\infty)$ and $\vartheta=A$, and let $A$ denote the latent Pareto shape parameter with prior law $\pi_\xi$.
The severities are conditionally independent given $(A,\xi)$ and, for $j=1,\ldots,n$, satisfy
\begin{equation*}
Y_j\mid(A=a,\xi)\sim\mathrm{Pareto}(x_{\min},a), \qquad f_a(y)=a x_{\min}^{a}y^{-(a+1)},\quad y\ge x_{\min}.
\end{equation*}
For observations $\bm{y}_n\in[x_{\min},\infty)^n$, the likelihood of $a$ is
\begin{equation*}
L_n(a;\bm{y}_n)=\prod_{j=1}^n a x_{\min}^{a}y_j^{-(a+1)}=\left(\prod_{j=1}^n y_j^{-1}\right)a^n\exp(-a t_n),\qquad t_n=\sum_{j=1}^n\log\left(\frac{y_j}{x_{\min}}\right).
\end{equation*}
Equivalently, the log-excesses $\log(Y_j/x_{\min})$ are conditionally independent exponential variables with rate $a$, so $T_n\coloneqq\sum_{j=1}^n\log(Y_j/x_{\min})$ satisfies $T_n\mid(A=a,\xi)\sim\mathrm{Gamma}(n,a)$ in shape--rate parameterization, and one may take $S_n=T_n$ and $s_n=t_n$.
As in the frequency model, this one-dimensional simplification is a consequence of the properties of the Pareto distribution and conditional independence, and the cWGAN formulation developed below is not restricted to this specification.
Moreover, while we treat $x_{\min}$ as fixed throughout, the method could incorporate threshold uncertainty by assigning it a prior, which we leave to future work.
We recall that the conditional mean is finite precisely when $a>1$, in which case $\E[Y_j\mid A=a,\xi]=a x_{\min}/(a-1)$, while the conditional variance is finite precisely when $a>2$, in which case $\Var(Y_j\mid A=a,\xi)=a x_{\min}^2/{(a-1)^2(a-2)}$.

For fixed $\xi$, the frequency model is a B\"uhlmann model, equivalently the special case of B\"uhlmann--Straub with unit exposure.
If $m_\xi=\E_{\pi_\xi}[\Lambda]<\infty$ and $0<\Var_{\pi_\xi}(\Lambda)<\infty$, then $\E[Y_j\mid\Lambda,\xi]=\Var(Y_j\mid\Lambda,\xi)=\Lambda$, and the B\"uhlmann predictor of the conditional mean is
\begin{equation*}
\widehat{\Lambda}_{n,\xi}^{\mathrm B}=Z_{n,\xi}\frac{Y_+}{n}+(1-Z_{n,\xi})m_\xi,\qquad Z_{n,\xi}=\frac{n}{n+m_\xi/\Var_{\pi_\xi}(\Lambda)}.
\end{equation*}
This is the mean-square optimal linear predictor of $\Lambda$ based on $Y_1,\ldots,Y_n$ which, in the Poisson--Gamma case, also equals $\E[\Lambda\mid\bm{Y}_n,\xi]$, although this identity need not hold for a general prior.
For the Pareto severity model, a classical B\"uhlmann premium would instead require $\E_{\pi_\xi}[\Var(Y_j\mid A,\xi)]<\infty$ and $\Var_{\pi_\xi}\{\E[Y_j\mid A,\xi]\}<\infty$.
These conditions are not satisfied by the Gamma, inverse-Gaussian, and lognormal severity priors used below: each assigns positive probability to $A\le2$, where the conditional second moment is infinite, and to $A\le1$, where even the conditional mean is infinite.
Accordingly, the severity analysis below is formulated in terms of the posterior law of $A$, rather than a B\"uhlmann credibility premium for the conditional mean $\E[Y_j\mid A,\xi]$.

For either of the two models above, let $L_{n,\xi}(\vartheta;\bm y_n)$ denote the corresponding likelihood conditional on $\xi$.
For fixed $\xi$, $n$, and a realized history $\bm y_n$, define the marginal likelihood
\begin{equation*}
M_{n,\xi}(\bm y_n)=\int_\Theta L_{n,\xi}(\vartheta;\bm y_n)\,\pi_\xi(\mathrm{d}\vartheta),
\end{equation*}
and suppose that $0<M_{n,\xi}(\bm y_n)<\infty$.
By Bayes' formula, the posterior law of $\vartheta$ is then the probability measure $P_{n,\xi}(\,\cdot\mid\bm y_n)$ given by
\begin{equation}\label{eq:posterior-law}
P_{n,\xi}(B\mid\bm y_n)\propto\int_B L_{n,\xi}(\vartheta;\bm y_n)\,\pi_\xi(\mathrm{d}\vartheta),\qquad B\in\mathcal B(\Theta),
\end{equation}
with normalizing constant $M_{n,\xi}(\bm y_n)$.
For the frequency and severity models, the corresponding posterior densities with respect to $\pi_\xi$ are therefore proportional, as functions of $\lambda$ and $a$, respectively, to $e^{-n\lambda}\lambda^{y_+}$ and $a^n e^{-a t_n}$.
Since the posterior must be updated across observation counts, claims histories, and prior specifications, the computational problem becomes to learn the map $(n,\bm y_n,\xi)\mapsto P_{n,\xi}(\,\cdot\mid\bm y_n)$.

Although the latent parameter is scalar in each of the two models (frequency and severity) considered here, the posterior law of $\vartheta$ identified in \cref{eq:posterior-law} does depend on the conditioning information $(n,\bm y_n,\xi)$.
Whenever neither the normalizing constant $M_{n,\xi}(\bm y_n)$ nor a direct posterior sampler is available in closed form, numerical quadrature or iterative simulation must be repeated for each conditioning value.
Let $K_\xi$ be a probability kernel from $(\Theta,\mathcal B(\Theta))$ to $(\R_+,\mathcal B(\R_+))$, and write $K_{\vartheta,\xi}=K_\xi(\vartheta,\cdot)$ for the conditional law of a future loss given $(\vartheta,\xi)$.
The posterior predictive law is then
\begin{equation*}
\Pi_{n,\xi}(B\mid\bm y_n)=\int_\Theta K_{\vartheta,\xi}(B)\,P_{n,\xi}(\mathrm{d}\vartheta\mid\bm y_n),\qquad B\in\mathcal B(\R_+).
\end{equation*}
Since this mixture is, in general, not determined by posterior summaries such as the posterior mean or variance alone, approximating the full posterior law provides a direct way to compute quantities such as predictive quantiles and tail probabilities.

It is well known that, if the prior belongs to a family conjugate to the likelihood, then the posterior belongs to the same family, which allows for efficient inference and prediction.
We recall this elementary property in the following lemma.
\begin{lemma}\label{lem:conjugacy}
Fix $\xi$ and $n\in\mathbb N$, and let $\{\pi_\zeta:\zeta\in H\}$ be a family of probability measures on $\Theta$ such that $\pi_\xi=\pi_{\zeta_0}$ for some $\zeta_0\in H$.
Suppose that the observations are conditionally i.i.d.\ given $(\vartheta,\xi)$, with density $f_{\vartheta,\xi}$, and that, for every $\zeta\in H$ and every observation $y$,
\begin{equation}\label{eq:conjugacy}
f_{\vartheta,\xi}(y)\,\pi_\zeta(\mathrm{d}\vartheta)=\kappa_{\xi,\zeta}(y)\,\pi_{U_\xi(\zeta,y)}(\mathrm{d}\vartheta),
\end{equation}
where $\kappa_{\xi,\zeta}(y)\in(0,\infty)$ and $U_\xi(\zeta,y)\in H$.
Set $\zeta_j=U_\xi(\zeta_{j-1},y_j)$, $j=1,\ldots,n$.
Then
\begin{equation*}
P_{n,\xi}(\,\cdot\mid\bm y_n)=\pi_{\zeta_n},\qquad M_{n,\xi}(\bm y_n)=\prod_{j=1}^n\kappa_{\xi,\zeta_{j-1}}(y_j).
\end{equation*}
\end{lemma}

\begin{proof}
Iterating the identity in \cref{eq:conjugacy} and integrating over $\Theta$ gives
\begin{equation*}
M_{n,\xi}(\bm y_n)=\int_{\Theta} L_{n,\xi}(\vartheta;\bm y_n)\,\pi_\xi(\mathrm{d}\vartheta)=\left\{\prod_{j=1}^n\kappa_{\xi,\zeta_{j-1}}(y_j)\right\}\int_{\Theta} \pi_{\zeta_n}(\mathrm{d}\vartheta)=\prod_{j=1}^n\kappa_{\xi,\zeta_{j-1}}(y_j), 
\end{equation*}
and normalization yields the posterior law.
\end{proof}
For the two benchmark models, taking $\pi_\xi=\mathrm{Gamma}(\alpha_\xi,\beta_\xi)$, $\alpha_\xi,\beta_\xi>0$, in shape--rate form and applying \cref{eq:conjugacy} gives the posterior laws 
\begin{equation*}
\Lambda\mid(\bm Y_n=\bm y_n,\xi)\sim\mathrm{Gamma}(\alpha_\xi+y_+,\beta_\xi+n),\qquad A\mid(\bm Y_n=\bm y_n,\xi)\sim\mathrm{Gamma}(\alpha_\xi+n,\beta_\xi+t_n).
\end{equation*}

However, not all cases are conjugate.
Suppose, for example, that the Gamma priors in the Poisson and Pareto models are replaced by lognormal priors,
\begin{equation*}
\log\Lambda\mid\xi\sim\mathcal N(\mu_{\Lambda,\xi},\sigma_{\Lambda,\xi}^2),\qquad \log A\mid\xi\sim\mathcal N(\mu_{A,\xi},\sigma_{A,\xi}^2),
\end{equation*}
where, for either prior, mean $m>0$ and coefficient of variation $v>0$ correspond to $\sigma^2=\log(1+v^2)$ and $\mu=\log m-\sigma^2/2$.
Putting $u=\log\lambda$ in the frequency model and $u=\log a$ in the severity model, the posterior densities given $(\bm Y_n=\bm y_n,\xi)$ with respect to Lebesgue measure on $\R$ are proportional, respectively, to
\begin{equation*}
\exp\!\left\{y_+u-ne^u-\frac{(u-\mu_{\Lambda,\xi})^2}{2\sigma_{\Lambda,\xi}^2}\right\},
\qquad
\exp\!\left\{nu-t_ne^u-\frac{(u-\mu_{A,\xi})^2}{2\sigma_{A,\xi}^2}\right\},
\qquad u\in\R.
\end{equation*}
In the frequency model, and in the severity model whenever $t_n>0$, the nonlinear $e^u$ terms prevent these densities from belonging to the Gaussian family.
Thus the lognormal prior is not conjugate in either model, and posterior quantities must in general be determined by numerical integration or simulation for each conditioning value.

In such cases, one possible approximation is provided by variational Bayes.
For fixed $(n,\bm y_n,\xi)$, let $\mathcal Q_\xi$ be a family of probability measures $q\ll\pi_\xi$ such that every $q\in\mathcal Q_\xi$ satisfies $L_{n,\xi}(\cdot;\bm y_n)>0$ $q$-almost everywhere, $\int_\Theta \abs{\log L_{n,\xi}(\vartheta;\bm y_n)}\,q(\mathrm{d}\vartheta)<\infty$, and $\operatorname{KL}(q\|\pi_\xi)<\infty$, where $\operatorname{KL}$ denotes the reverse Kullback--Leibler divergence.
Variational Bayes then selects $q\in\mathcal Q_\xi$ by maximizing the evidence lower bound (ELBO)
\begin{equation*}
\mathcal E_{n,\xi}(q;\bm y_n)=\int_\Theta \log L_{n,\xi}(\vartheta;\bm y_n)\,q(\mathrm{d}\vartheta)-\operatorname{KL}(q\|\pi_\xi).
\end{equation*}
By \cref{eq:posterior-law}, $\log M_{n,\xi}(\bm y_n)-\mathcal E_{n,\xi}(q;\bm y_n)=\operatorname{KL}\!\left(q\,\middle\|\,P_{n,\xi}(\,\cdot\mid\bm y_n)\right)$.
Hence maximizing the ELBO over $\mathcal Q_\xi$ is equivalent to minimizing the reverse Kullback--Leibler divergence to the posterior law $P_{n,\xi}(\,\cdot\mid\bm y_n)$.
This optimization must, however, still be repeated separately for each $(n,\bm y_n,\xi)$, whereas an amortized parametrization $q_\varphi(\,\cdot\mid n,\bm y_n,\xi)\in\mathcal Q_\xi$ shares $\varphi$ across conditioning values.

The dimension of the conditioning variable may be reduced whenever the sampling family admits a finite-dimensional sufficient statistic.
We recall this property in the following lemma:
\begin{lemma}
Fix $\xi$ and $n,r\in\mathbb N$.
Suppose that, conditional on $(\vartheta,\xi)$, $Y_1,\ldots,Y_n$ are i.i.d.\ and their common law has, with respect to a common dominating measure, a density
\begin{equation*}
f_{\vartheta,\xi}(y)=h_\xi(y)\exp\!\left\{\eta_\xi(\vartheta)^\top t_\xi(y)-b_\xi(\vartheta)\right\},
\end{equation*}
where $h_\xi:\mathsf Y\to[0,\infty)$ and $t_\xi:\mathsf Y\to\R^r$ are $\mathcal Y$-measurable, while $\eta_\xi:\Theta\to\R^r$ and $b_\xi:\Theta\to\R$ are $\mathcal B(\Theta)$-measurable. Suppose further that the support of $f_{\vartheta,\xi}$ does not depend on $\vartheta$, and define $S_n=\sum_{j=1}^n t_\xi(Y_j)$.
Whenever $0<M_{n,\xi}(\bm y_n)<\infty$, with $s_n=\sum_{j=1}^n t_\xi(y_j)$, the posterior satisfies
\begin{equation*}
P_{n,\xi}(B\mid\bm y_n)\propto\int_B\exp\!\left\{\eta_\xi(\vartheta)^\top s_n-nb_\xi(\vartheta)\right\}\,\pi_\xi(\mathrm{d}\vartheta),\qquad B\in\mathcal B(\Theta).
\end{equation*}
Consequently, for fixed $(n,\xi)$, both the posterior and posterior predictive laws depend on $\bm y_n$ only through $s_n$, regardless of whether $\pi_\xi$ is conjugate or not.
\end{lemma}

\begin{proof}
The likelihood factors as $\prod_{j=1}^n h_\xi(y_j)\exp\{\eta_\xi(\vartheta)^\top s_n-nb_\xi(\vartheta)\}$.
The Fisher--Neyman factorization criterion therefore shows that $S_n$ is sufficient for $\vartheta$ (for fixed $\xi$).
Substitution into Equation~\eqref{eq:posterior-law} and cancellation of the factor $\prod_{j=1}^n h_\xi(y_j)$ yield the displayed posterior law, which shows that the posterior depends on $\bm y_n$ only through $s_n$.
Because $K_{\vartheta,\xi}$ does not depend on $\bm y_n$, plugging the posterior into the mixture formula $\Pi_{n,\xi}(B\mid\bm y_n)=\int_\Theta K_{\vartheta,\xi}(B)\,P_{n,\xi}(\mathrm d\vartheta\mid\bm y_n)$ shows that the posterior predictive law also depends on $\bm y_n$ only through $s_n$.
\end{proof}

For the Poisson and Pareto models, $S_n=Y_+$ and $S_n=T_n$, respectively.
Hence the conditioning variable $(n,\bm y_n,\xi)$ may be replaced by $c=(n,y_+,\xi)$ in the frequency model and by $c=(n,t_n,\xi)$ in the severity model.
Conversely, the Pitman--Koopman--Darmois theorem implies, under suitable regularity assumptions, that an i.i.d.\ family with parameter-independent support admitting, for every sample size, a sufficient statistic whose dimension is bounded independently of $n$ is an exponential family; see, e.g., \citet{koopman1936distributions}.

\section{Conditional Wasserstein generative adversarial networks}\label{sec:cwgan}
Approximating a target law by a `generated' law requires choosing a metric with respect to which the distance between the two laws can be formulated and optimized.
In this work, we consider the 1-Wasserstein distance.
Let $d\in\mathbb N$, and set
\begin{equation*}
\mathcal P_1(\R^d)=\left\{\mu\in\mathcal P(\R^d)\, \colon\, \int_{\R^d}\norm{x}\,\mu(\mathrm{d}x)<\infty\right\},
\end{equation*}
where $\mathcal P(\R^d)$ denotes the space of Borel probability measures on $\R^d$ and $\norm{\cdot}$ is the Euclidean norm.
For $\mu,\nu\in\mathcal P_1(\R^d)$, let $\Pi(\mu,\nu)$ denote the set of Borel probability measures on $\R^d\times\R^d$ with first marginal $\mu$ and second marginal $\nu$, and define
\begin{equation}\label{eq:w1dist}
W_1(\mu,\nu)=\inf_{\gamma\in\Pi(\mu,\nu)}\int_{\R^d\times\R^d}\norm{x-y}\,\gamma(\mathrm{d}x,\mathrm{d}y).
\end{equation}
Indeed, \eqref{eq:w1dist} is finite: since $\mu\otimes\nu\in\Pi(\mu,\nu)$, the triangle inequality gives
\begin{equation*}
W_1(\mu,\nu)\leq\int_{\R^d\times\R^d}\norm{x-y}\,\mu(\mathrm{d}x)\,\nu(\mathrm{d}y)\leq\int_{\R^d}\norm{x}\,\mu(\mathrm{d}x)+\int_{\R^d}\norm{y}\,\nu(\mathrm{d}y)<\infty.
\end{equation*}
$W_1$ is a metric on $\mathcal P_1(\R^d)$, and $(\mathcal P_1(\R^d),W_1)$ is a Polish space \citep[Theorem~6.18]{villani2009optimal}.

Since the primal formulation in Equation~\eqref{eq:w1dist} is often difficult to compute, we turn to its dual representation.
Let $\operatorname{Lip}_1(\R^d)$ denote the class of real-valued 1-Lipschitz functions on $\R^d$.
For $\mu,\nu\in\mathcal P_1(\R^d)$, every $f\in\operatorname{Lip}_1(\R^d)$ is integrable with respect to both measures, since $\abs{f(x)}\leq \abs{f(0)}+\norm{x}$.
The Kantorovich--Rubinstein duality \citep[Particular Case~5.16]{villani2009optimal} yields the following representation of the 1-Wasserstein distance:
\begin{equation}\label{eq:w1dual}
W_1(\mu,\nu)=\sup_{f\in\operatorname{Lip}_1(\R^d)}\left(\int_{\R^d} f(x)\,\mu(\mathrm{d}x)-\int_{\R^d} f(x)\,\nu(\mathrm{d}x)\right).
\end{equation}
For later use, we note that the Kantorovich--Rubinstein representation remains valid on any Polish metric space, not just on $\R^d$.
Because $\operatorname{Lip}_1(\R^d)$ is closed under $f\mapsto-f$, the same value is obtained by placing an absolute value around the difference of integrals.
Moreover, because additive constants cancel, the supremum is unchanged if one imposes the normalization $f(0)=0$.
\cref{eq:w1dual} is particularly useful, as it replaces the infimum over the somewhat abstract couplings $\gamma \in \Pi(\mu, \nu)$ by a supremum over the tractable set of 1-Lipschitz `test' functions.

Wasserstein generative adversarial networks (WGANs, \citet{arjovsky2017wasserstein}) represent candidate laws as pushforwards of a fixed noise law and restrict the class $\operatorname{Lip}_1(\R^d)$ in the Kantorovich--Rubinstein representation in \cref{eq:w1dual} to parametrized neural critics.
Let $Z$ have a fixed probability law $\nu_Z$ on $\R^{d_z}$, and let the nonempty set $\mathsf G$ index measurable neural generator networks $G_\theta:\R^{d_z}\to\R^d$.
For a target law $\mu\in\mathcal P_1(\R^d)$, we write $\mu_\theta=(G_\theta)_\#\nu_Z$ for the pushforward of $\nu_Z$ under $G_\theta$ and assume $\mu_\theta\in\mathcal P_1(\R^d)$ for every $\theta\in\mathsf G$.
For a parametrized family of scalar-valued neural critics $D_\phi:\R^d\to\R$ containing at least one 1-Lipschitz critic, the WGAN objective is then given by
\begin{equation}\label{eq:w1nn}
\inf_{\theta\in\mathsf G}\sup_{\substack{\phi:\,D_\phi\in\operatorname{Lip}_1(\R^d)}}\Bigg(\int_{\R^d} D_\phi(x)\,\mu(\mathrm{d}x) -\int_{\R^{d_z}}D_\phi(G_\theta(z))\,\nu_Z(\mathrm{d}z)\Bigg).
\end{equation}
We note that, for fixed $\theta$, the inner supremum in \cref{eq:w1nn} is bounded above by $W_1(\mu,\mu_\theta)$ and need not be equal to it, since the critic class is restricted.
Likewise, the generator family restricts the attainable pushforward laws.
Since both $D_\phi$ and $G_\theta$ are parametrized by neural networks, we can use stochastic gradient descent to optimize the parameters $\theta$ and $\phi$ in an alternating fashion.
For completeness, we briefly recall the definition of shallow feedforward neural networks in the following definition.
\begin{definition}
Let $\psi \colon \R \to \R$ be a measurable activation function.
We denote by $\mathcal{NN}(\psi)$ the set of shallow feedforward neural networks with activation function $\psi$, which is the $\mathbb{R}$-linear span of functions of the form $x \mapsto \psi(\langle w, x \rangle + b)$, where the input dimension $q\in\mathbb{N}$ is implicitly understood from the context, and $w \in \mathbb{R}^q$ and $b \in \mathbb{R}$ are the weight vector and bias term, respectively.
In other words,
\begin{equation}
\mathcal{NN}(\psi) = \mathrm{span}\left\{ x \mapsto \psi(\langle w, x \rangle + b) \mid w \in \mathbb{R}^q, b \in \mathbb{R} \right\}.
\end{equation}
\end{definition}

Equation~\eqref{eq:w1dual} identifies the 1-Wasserstein distance as an integral probability metric \citep{mueller1997integral}.
Given a nonempty class $\mathcal F$ of real-valued Borel functions on $\R^d$, define, on any subclass of $\mathcal P(\R^d)$ for which the integrals and the following supremum are finite,
\begin{equation*}
d_{\mathcal F}(\mu,\nu)=\sup_{f\in\mathcal F}\left|\int_{\R^d} f(x)\,\mu(\mathrm{d}x)-\int_{\R^d} f(x)\,\nu(\mathrm{d}x)\right|.
\end{equation*}
Then $d_{\mathcal F}$ is a pseudometric and is a metric if $\mathcal F$ separates the probability measures under consideration.
In particular, taking $\mathcal F=\operatorname{Lip}_1(\R^d)$ recovers $W_1$ on $\mathcal P_1(\R^d)$.
With the convention $d_{\mathrm{TV}}(\mu,\nu)\coloneqq\sup_{A\in\mathcal B(\R^d)}\abs{\mu(A)-\nu(A)}$, taking $\mathcal F=\{\mathbf 1_A:A\in\mathcal B(\R^d)\}$ yields total variation distance, while, for $d=1$, the class $\mathcal F=\{\mathbf 1_{(-\infty,t]}:t\in\R\}$ yields the Kolmogorov distance.
If $\mathcal H$ denotes a reproducing kernel Hilbert space of functions on $\R^d$, then $\mathcal F=\{f\in\mathcal H:\norm{f}_{\mathcal H}\leq1\}$ yields the maximum mean discrepancy, whereas $\mathcal F_{\mathrm{BL}}=\{f\in\operatorname{Lip}_1(\R^d):\norm{f}_\infty\leq1\}$ yields the bounded-Lipschitz metric, which is finite on $\mathcal P(\R^d)$ and metrizes weak convergence.
Another particularly interesting and instructive case is the Cramér distance.
For $d=1$, let $C_{\mathrm{ac}}(\R)$ denote the class of absolutely continuous functions on $\R$.
Taking $\mathcal F=\{f\in C_{\mathrm{ac}}(\R):\norm{f'}_{L^2(\R)}\leq1\}$ yields, for $\mu,\nu\in\mathcal P_1(\R)$,
\begin{equation*}
d_{\mathcal F}(\mu,\nu)=\left(\int_{\R}\left(F_\mu(x)-F_\nu(x)\right)^2\,\mathrm{d}x\right)^{1/2},
\end{equation*}
where $F_\mu$ and $F_\nu$ denote the cumulative distribution functions of $\mu$ and $\nu$, respectively.
Consequently, the framework of GANs can be generalized to other integral probability metrics, where the discriminator network is trained to maximize a different objective function corresponding to a specific choice of $\mathcal F$.
By contrast, $f$-divergences do not, in general, arise as integral probability metrics from a choice of $\mathcal F$; we refer the interested reader to \citet{nowozin2016fgan}.

Under suitable assumptions on the activation function $\psi$, $\mathcal{NN}(\psi)$ is dense in $L^p(m)$ for every $1<p<\infty$ and every finite Borel measure $m$ on $\R^d$, and the restriction of $\mathcal{NN}(\psi)$ to any compact set $K\subset\R^d$ is dense in $C(K)$ under the uniform norm; see, e.g., \citet{hornik1991approximation}.
However, these classical density results do not, by themselves, control the approximation errors in~\eqref{eq:w1nn}.
To relate function approximation to approximation of generated laws, suppose that a measurable map $T:\R^{d_z}\to\R^d$ satisfies $T_\#\nu_Z=\mu$.
Then $T(Z)$ and $G_\theta(Z)$ are a coupling of $\mu$ and $\mu_\theta$, i.e., $(T,G_\theta)_\#\nu_Z\in\Pi(\mu,\mu_\theta)$, and hence 
\begin{equation*}
W_1(\mu,\mu_\theta)\leq\int_{\R^{d_z}}\norm{T(z)-G_\theta(z)}\,\nu_Z(\mathrm{d}z).
\end{equation*}
Hence, assuming that $T_\#\nu_Z=\mu$ for some measurable map $T$, $L^1(\nu_Z)$-approximation of $T$ is sufficient for $W_1$-approximation of $\mu$.
More generally, for WGANs, one must control both the generator error $\inf_{\theta\in\mathsf G}W_1(\mu,\mu_\theta)$ and the error incurred by restricting the supremum in~\eqref{eq:w1dual} to the neural critics $D_\phi\in\operatorname{Lip}_1(\R^d)$.

The $1$-Wasserstein criterion underlying the WGAN objective in Equation~\eqref{eq:w1nn} may be extended to the conditional setting as follows.
Let $(C,\mathcal C)$ be a measurable space of conditioning values containing the values $c=(n,s_n,\xi)$ identified in \cref{sec:problem_formulation}, let $\rho$ be their law, and assume that $c\mapsto P_c$ is a probability kernel from $(C,\mathcal C)$ to $(\Theta,\mathcal B(\Theta))$.
For $\theta\in\mathsf G$, let $G_\theta\colon C\times\R^{d_z}\to\Theta$ be a measurable conditional generator, and define $Q_{\theta,c}=\bigl(G_\theta(c,\cdot)\bigr)_\#\nu_Z$.
For each fixed $\theta\in\mathsf G$, the family $c\mapsto Q_{\theta,c}$ is a probability kernel, since, for every $B\in\mathcal B(\Theta)$,
\begin{equation*}
Q_{\theta,c}(B)=\int_{\R^{d_z}}\mathbf 1_B\!\left(G_\theta(c,z)\right)\nu_Z(\mathrm dz),
\end{equation*}
and the right-hand side is measurable in $c$ by the joint measurability of $G_\theta$.
As $\Theta\subseteq\R^d$, we regard $P_c$ and $Q_{\theta,c}$ as Borel probability measures on $\R^d$.
Assume, for every $\theta\in\mathsf G$, that $P_c,Q_{\theta,c}\in\mathcal P_1(\R^d)$ for every $c\in C$ and $\int_C\int_\Theta \norm{\vartheta}\left\{P_c(\mathrm d\vartheta)+Q_{\theta,c}(\mathrm d\vartheta)\right\}\rho(\mathrm dc)<\infty$.
The conditional Wasserstein objective is to minimize over $\theta\in\mathsf G$ the integrated conditional error
\begin{equation}\label{eq:integrated-w1}
\mathcal R(\theta)=\int_C W_1\!\left(P_c,Q_{\theta,c}\right)\rho(\mathrm dc).
\end{equation}
Let $\mathcal D_1$ be the class of jointly measurable functions $D\colon C\times\R^d\to\R$ such that, for all $c\in C$, $D(c, \cdot) \in \operatorname{Lip}_1(\R^d)$.
The proof of \cref{prop:conditional-kpack-identity}, specialized to $k=1$, yields the representation
\begin{equation}\label{eq:jointlaw}
\mathcal R(\theta)=\sup_{D\in\mathcal D_1}\int_{C\times\Theta\times\R^{d_z}}\left\{D(c,\vartheta)-D\!\left(c,G_\theta(c,z)\right)\right\}\rho(\mathrm dc)P_c(\mathrm d\vartheta)\nu_Z(\mathrm dz).
\end{equation}
Thus \cref{eq:jointlaw} replaces the integral of pointwise Wasserstein suprema $W_1(P_c,Q_{\theta,c})$ by a single supremum over jointly measurable critics $D\in\mathcal D_1$ and a single integral with respect to the joint law $\rho(\mathrm dc)P_c(\mathrm d\vartheta)\nu_Z(\mathrm dz)$, which is particularly convenient in hierarchical Bayesian models, where samples from $\rho(\mathrm dc)P_c(\mathrm d\vartheta)$ may be generated directly from the hierarchical model.
We note that the condition $D(c,\cdot)\in\operatorname{Lip}_1(\R^d)$ for every $c\in C$ does not by itself imply joint measurability of $(c,\vartheta)\mapsto D(c,\vartheta)$.
The proof of \cref{prop:conditional-kpack-identity} shows, however, that restricting the supremum in \cref{eq:jointlaw} to jointly measurable critics entails no loss.

While the formulation in \cref{eq:jointlaw} is convenient, its direct empirical implementation has a potential drawback.
In practice, the objective in \cref{eq:jointlaw} is estimated from independent samples $(C_i,\vartheta_i)$ with joint law $\rho(\mathrm{d}c)P_c(\mathrm{d}\vartheta)$.
When $\rho$ is non-atomic, the conditioning values in a finite batch are almost surely distinct.
Even when $\rho$ has atoms, repetitions may be too sparse for the critic to observe conditional variation at a fixed $c$ directly.
With a restricted neural critic class and imperfect optimization, the critic may therefore fail to detect loss of conditional variation in the generated law, and mode collapse may occur.
In the extreme case, $Q_{\theta,c}$ may be a Dirac measure for every $c\in C$.
Although Wasserstein GANs were introduced in part to mitigate mode collapse in classical GANs \citep{arjovsky2017wasserstein}, the phenomenon may nevertheless persist in the conditional setting.

A practical workaround to this issue may be obtained as follows.
For $\theta\in\mathsf G$, $c\in C$, and $k\in\mathbb N$, define the hybrid $k$-pack law on $(\R^d)^k$ by
\begin{equation}\label{eq:kpack}
H_{\theta,c}^{(k)}\coloneqq\frac{1}{k}\sum_{j=1}^kQ_{\theta,c}^{\otimes(j-1)}\otimes P_c\otimes Q_{\theta,c}^{\otimes(k-j)}.
\end{equation}
Here and below, tensor factors with exponent $0$ are omitted; in particular,
$H_{\theta,c}^{(1)}=P_c$.
Equivalently, one first chooses a coordinate uniformly from $\{1,\ldots,k\}$ and, conditional on this choice, draws the coordinates independently, with law $P_c$ at the selected coordinate and $Q_{\theta,c}$ at the remaining coordinates.
We adopt the term \emph{pack} from PacGAN \citep{lin2018pacgan}, where the discriminator acts jointly on several independent samples from a common law.
By contrast, the hybrid pack in \cref{eq:kpack} consists of one draw from $P_c$ and $k-1$ draws from $Q_{\theta,c}$.
The special case $k=2$ was proposed by \citet{adler2025deep}; however, the construction in \cref{eq:kpack} is more general and allows for arbitrary $k\geq1$.

For $k\geq2$, a critic distinguishing between draws from $H_{\theta,c}^{(k)}$ and $Q_{\theta,c}^{\otimes k}$ rather than from $P_c$ and $Q_{\theta,c}$ directly is evaluated jointly on several draws corresponding to the same conditioning value $c$.
However, a priori, training with hybrid $k$-packs may not be equivalent to training with the original conditional Wasserstein objective in \cref{eq:integrated-w1}.
\cref{prop:conditional-kpack-identity} below shows that the two objectives differ only by the factor $1/k$ (which is an artifact of the choice of metric on the product space) and, consequently, have the same minimizing sequences and, whenever minimizers exist, the same minimizers.
In order to see this, for every $k\geq1$, we equip $(\R^d)^k$ with the norm $\norm{\boldsymbol\vartheta}_k=\frac1k\sum_{i=1}^k\norm{\vartheta_i}$.
For notational simplicity, we suppress the dependence of $H_{\theta,c}^{(k)}$ and $Q_{\theta,c}^{\otimes k}$ on $\theta$ and write $H_c^{(k)}$ and $Q_c^{\otimes k}$ instead.
Moreover, we convene that, when applied to probability measures on $(\R^d)^k$, $W_1$ denotes the $1$-Wasserstein distance induced by $\norm{\cdot}_k$.
Let $\mathcal D_k$ be the class of all $\mathcal C\otimes\mathcal B((\R^d)^k)$-measurable functions $D:C\times(\R^d)^k\to\R$ such that, for every $c\in C$, $D(c,\cdot)\in\operatorname{Lip}_1((\R^d)^k)$ with respect to the norm $\norm{\cdot}_k$.
We define
\begin{equation*}
\mathcal J_k(\theta)=\sup_{D\in\mathcal D_k}\int_C\left\{\int_{(\R^d)^k}D(c,\boldsymbol\vartheta)\,H_c^{(k)}(\mathrm d\boldsymbol\vartheta)-\int_{(\R^d)^k}D(c,\boldsymbol\vartheta)\,Q_c^{\otimes k}(\mathrm d\boldsymbol\vartheta)\right\}\rho(\mathrm dc).
\end{equation*}

\begin{proposition}\label{prop:conditional-kpack-identity}
For every $\theta\in\mathsf G$ and every integer $k \geq 1$,
\begin{equation*}
\mathcal J_k(\theta)=\int_C W_1\!\left(H_c^{(k)},Q_c^{\otimes k}\right)\rho(\mathrm dc)=\frac1k\mathcal R(\theta).
\end{equation*}
\end{proposition}

\begin{proof}
In order to see that, for fixed $c\in C$, both $Q_c^{\otimes k}$ and $H_c^{(k)}$ have finite first moment with respect to $\norm{\cdot}_k$, set $m_P(c)=\int_\Theta \norm{\vartheta}P_c(\mathrm d\vartheta)$ and $m_Q(c)=\int_\Theta \norm{\vartheta}Q_c(\mathrm d\vartheta)$, and note that $m_P(c)$ and $m_Q(c)$ are finite by assumption.
Moreover, for every $c\in C$,
\begin{equation*}
\int_{(\R^d)^k}\norm{\boldsymbol\vartheta}_k H_c^{(k)}(\mathrm d\boldsymbol\vartheta)=\frac1k m_P(c)+\left(1-\frac1k\right)m_Q(c),
\quad
\int_{(\R^d)^k}\norm{\boldsymbol\vartheta}_k Q_c^{\otimes k}(\mathrm d\boldsymbol\vartheta)=m_Q(c).
\end{equation*}

We first show the pointwise identity
\begin{equation}\label{eq:pointwise-kpack-identity}
W_1(H_c^{(k)},Q_c^{\otimes k})=\frac1k W_1(P_c,Q_c),\qquad c\in C.
\end{equation}
Fix $c\in C$, and abbreviate $P=P_c$, $Q=Q_c$, and $H=H_c^{(k)}$.
Let $\gamma\in\Pi(P,Q)$, and fix $j\in\{1,\ldots,k\}$.
Let $(U,V)$ have law $\gamma$, and, independently of $(U,V)$, let $X_i\sim Q$, $i\neq j$, be mutually independent.
The joint law of $(\boldsymbol X^{(j)},\boldsymbol Y^{(j)})$, where
\begin{equation*}
\boldsymbol X^{(j)} =(X_1,\ldots,X_{j-1},U,X_{j+1},\ldots,X_k),\qquad \boldsymbol Y^{(j)}=(X_1,\ldots,X_{j-1},V,X_{j+1},\ldots,X_k),
\end{equation*}
is a coupling of $Q^{\otimes(j-1)}\otimes P\otimes Q^{\otimes(k-j)}$ and $Q^{\otimes k}$.
Averaging these $k$ couplings gives a coupling $\overline{\Gamma}\in\Pi(H,Q^{\otimes k})$ with expected cost
\begin{equation*}
\int_{(\R^d)^k\times(\R^d)^k}\norm{\boldsymbol\vartheta-\boldsymbol\vartheta'}_k\,\overline{\Gamma}(\mathrm d\boldsymbol\vartheta,\mathrm d\boldsymbol\vartheta')=\frac1k\int_{\R^d\times\R^d}\norm{\vartheta-\vartheta'}\,\gamma(\mathrm d\vartheta,\mathrm d\vartheta').
\end{equation*}
Taking the infimum over $\gamma\in\Pi(P,Q)$ yields $W_1(H,Q^{\otimes k})\leq\frac1k W_1(P,Q)$.

Conversely, for $f\in\operatorname{Lip}_1(\R^d)$, define $F(\boldsymbol\vartheta)=\frac1k\sum_{i=1}^k f(\vartheta_i)$.
Then $F$ is $1$-Lipschitz with respect to $\norm{\cdot}_k$.
Applying \cref{eq:w1dual} on $\bigl((\R^d)^k,\norm{\cdot}_k\bigr)$ yields
\begin{align*}
W_1(H,Q^{\otimes k})=\sup_{G\in\operatorname{Lip}_1((\R^d)^k)}&\left\{\int_{(\R^d)^k}G\,\mathrm dH-\int_{(\R^d)^k}G\,\mathrm dQ^{\otimes k}\right\}\\ \geq{}&\int_{(\R^d)^k}F\,\mathrm dH-\int_{(\R^d)^k}F\,\mathrm dQ^{\otimes k}={}\frac1k\left\{\int_{\R^d} f\,\mathrm dP-\int_{\R^d} f\,\mathrm dQ\right\}.
\end{align*}
Taking the supremum of the last expression over all $f\in\operatorname{Lip}_1(\R^d)$ and applying \cref{eq:w1dual} gives the reverse inequality and proves~\eqref{eq:pointwise-kpack-identity}.

Write $\mathscr L_k = \operatorname{Lip}_1((\R^d)^k)$.
Since $\bigl((\R^d)^k,\norm{\cdot}_k\bigr)$ is separable, there exists a countable dense set $\{\boldsymbol x_n:n\geq1\}\subset(\R^d)^k$, and we define
\begin{equation*}
r(F,G)=\sum_{n=1}^{\infty}2^{-n}\left(|F(\boldsymbol x_n)-G(\boldsymbol x_n)|\wedge1\right),\qquad F,G\in\mathscr L_k.
\end{equation*}
Lipschitz functions are continuous and are therefore determined by their values on a dense set, hence $r$ is a metric on $\mathscr L_k$.
It follows that the map $T:F\mapsto(F(\boldsymbol x_n))_{n\geq1}$ is injective and identifies $(\mathscr L_k,r)$ with a subspace $T(\mathscr L_k)$ of $\mathbb R^{\mathbb N}$ equipped with the topology of coordinatewise convergence (since $r$ metrizes coordinatewise convergence on $\mathscr L_k$, i.e. $r(F_j,F)\to 0 \iff F_j(\boldsymbol x_n)\to F(\boldsymbol x_n)$ for every $n$).
Since $\mathbb R^{\mathbb N}$ is second countable in the product topology (being a countable product of second-countable spaces), so is this subspace, and hence $(\mathscr L_k,r)$ is separable.
Thus there exists an $r$-dense sequence $(F_m)_{m\geq1}\subset\mathscr L_k$.
The Kantorovich--Rubinstein duality (\cref{eq:w1dual}) gives
\begin{equation*}
\sup_{m\geq1}\left\{\int_{(\R^d)^k}F_m\,\mathrm d\mu-\int_{(\R^d)^k}F_m\,\mathrm d\nu\right\}\leq W_1(\mu,\nu),\qquad\mu,\nu\in\mathcal P_1((\R^d)^k).
\end{equation*}

To prove the converse inequality, fix $F\in\mathscr L_k$.
By density, there is a sequence $(F_{m_\ell})_{\ell\geq1}$ selected from $(F_m)_{m\geq1}$ such that $r(F_{m_\ell},F)\to0$.
Thus $F_{m_\ell}(\boldsymbol x_n)\to F(\boldsymbol x_n)$ for every $n$.
For arbitrary $\boldsymbol\vartheta\in(\R^d)^k$ and every $n$, the common Lipschitz bound gives
\begin{align*}
|F_{m_\ell}(\boldsymbol\vartheta)-F(\boldsymbol\vartheta)|&\leq|F_{m_\ell}(\boldsymbol\vartheta)-F_{m_\ell}(\boldsymbol x_n)|+|F_{m_\ell}(\boldsymbol x_n)-F(\boldsymbol x_n)|+|F(\boldsymbol x_n)-F(\boldsymbol\vartheta)|\\
&\leq 2\norm{\boldsymbol\vartheta-\boldsymbol x_n}_k+|F_{m_\ell}(\boldsymbol x_n)-F(\boldsymbol x_n)|.
\end{align*}
Choosing first $\boldsymbol x_n$ close to $\boldsymbol\vartheta$ and then letting $\ell\to\infty$ shows that $F_{m_\ell}(\boldsymbol\vartheta)\to F(\boldsymbol\vartheta)$.
Since $F_{m_\ell}(0)\to F(0)$, the sequence $(F_{m_\ell}(0))_{\ell\geq1}$ is bounded.
Hence there exists $C<\infty$ such that, using the Lipschitz property, $|F_{m_\ell}(\boldsymbol\vartheta)|\leq C+\norm{\boldsymbol\vartheta}_k$, for all $\ell\geq1$.
Therefore, for any $\mu,\nu\in\mathcal P_1((\R^d)^k)$, dominated convergence yields
\begin{align*}
\int_{(\R^d)^k}F\,\mathrm d\mu-\int_{(\R^d)^k}F\,\mathrm d\nu &=\lim_{\ell\to\infty}\left\{\int_{(\R^d)^k}F_{m_\ell}\,\mathrm d\mu-\int_{(\R^d)^k}F_{m_\ell}\,\mathrm d\nu\right\}\\
&\leq\sup_{m\geq1}\left\{\int_{(\R^d)^k}F_m\,\mathrm d\mu-\int_{(\R^d)^k}F_m\,\mathrm d\nu\right\}.
\end{align*}
Taking the supremum over $F\in\mathscr L_k$ and applying the Kantorovich--Rubinstein duality yields the reverse inequality, and therefore
\begin{equation*}
W_1(\mu,\nu)=\sup_{m\geq1}\left\{\int_{(\R^d)^k}F_m\,\mathrm d\mu-\int_{(\R^d)^k}F_m\,\mathrm d\nu\right\},\qquad\mu,\nu\in\mathcal P_1((\R^d)^k).
\end{equation*}

Set
\begin{equation*}
a_m(c)=\int_{(\R^d)^k}F_m\,\mathrm dH_c^{(k)}-\int_{(\R^d)^k}F_m\,\mathrm dQ_c^{\otimes k},\qquad w(c)=\sup_{m\geq1}a_m(c).
\end{equation*}
Since $F_m$ is Borel measurable and $H_c^{(k)}$ and $Q_c^{\otimes k}$ are probability kernels, $a_m$ is $\mathcal C$-measurable for every $m$, and therefore $w$ is $\mathcal C$-measurable.
The preceding arguments imply that $w(c)=W_1(H_c^{(k)},Q_c^{\otimes k})$, and~\eqref{eq:pointwise-kpack-identity} yields $0\leq w(c)=k^{-1}W_1(P_c,Q_c)\leq k^{-1}\{m_P(c)+m_Q(c)\}$, so $w\in L^1(\rho)$, and the Kantorovich--Rubinstein duality implies, for every $D\in\mathcal D_k$, that the absolute value of its pointwise difference of expectations is bounded by $w(c)$, hence $\mathcal J_k(\theta)\leq\int_Cw\,\mathrm d\rho$.
For the converse inequality, define $m_\varepsilon(c)=\min\{m\geq1:a_m(c)>w(c)-\varepsilon\}$ for $\varepsilon>0$, and note that $m_\varepsilon:C\to\mathbb N$ is measurable, and that $D_\varepsilon(c,\boldsymbol\vartheta)=F_{m_\varepsilon(c)}(\boldsymbol\vartheta)$ belongs to $\mathcal D_k$.
Therefore
\begin{equation*}
\mathcal J_k(\theta)\geq\int_C a_{m_\varepsilon(c)}(c)\,\rho(\mathrm dc)\geq\int_Cw(c)\,\rho(\mathrm dc)-\varepsilon,
\end{equation*}
and letting $\varepsilon\downarrow0$ proves the first equality for $\mathcal J_k(\theta)$.
Integrating~\eqref{eq:pointwise-kpack-identity} with respect to $\rho$ then gives the second equality.
\end{proof}

\begin{remark}\label{rem:permutation-invariant-kpack-critics}
By \cref{eq:kpack}, permuting the pack coordinates only changes the order of the $k$ terms in $H_c^{(k)}$.
The product law $Q_c^{\otimes k}$ is also unchanged.
Thus both laws are exchangeable for every $c\in C$.
Hence the value of $\mathcal J_k(\theta)$ does not change if we restrict $\mathcal D_k$ to critics that are permutation-invariant.
Let $\mathfrak S_k$ be the symmetric group.
For $D\in\mathcal D_k$, define
\begin{equation*}
\tau_\sigma(\boldsymbol\vartheta)=(\vartheta_{\sigma(1)},\ldots,\vartheta_{\sigma(k)}),\qquad \overline D(c,\boldsymbol\vartheta)=\frac{1}{k!}\sum_{\sigma\in\mathfrak S_k}D\bigl(c,\tau_\sigma(\boldsymbol\vartheta)\bigr).
\end{equation*}
Each $\tau_\sigma$ is an isometry of $\bigl((\R^d)^k,\norm{\cdot}_k\bigr)$.
It follows that $\overline D$ is jointly measurable and $1$-Lipschitz with respect to $\norm{\cdot}_k$, and it is permutation invariant by construction.
Thus $\overline D\in\mathcal D_k$.
For $\mu_c=H_c^{(k)}$ or $\mu_c=Q_c^{\otimes k}$, exchangeability gives
\begin{align*}
\int_{(\R^d)^k}\overline D(c,\boldsymbol\vartheta)\,\mu_c(\mathrm d\boldsymbol\vartheta)=\frac{1}{k!}\sum_{\sigma\in\mathfrak S_k}\int_{(\R^d)^k}D\bigl(c,\tau_\sigma(\boldsymbol\vartheta)\bigr)\,\mu_c(\mathrm d\boldsymbol\vartheta)=\int_{(\R^d)^k}D(c,\boldsymbol\vartheta)\,\mu_c(\mathrm d\boldsymbol\vartheta).
\end{align*}
Permutation invariance may, for example, be imposed by means of a DeepSets architecture \citep{zaheer2017deep}.
\end{remark}

Before we conclude this section, we present a general stability result for posterior predictive laws.
To this end, let $K$ be a probability kernel from $(C\times\Theta,\mathcal C\otimes\mathcal B(\Theta))$ to $(\R_+,\mathcal B(\R_+))$, write $K_{c,\vartheta}=K((c,\vartheta),\cdot)$, and assume that $K_{c,\vartheta}\in\mathcal P_1(\R_+)$ for every $(c,\vartheta)$.
With $m(c,\vartheta)=\int_{\R_+}s\,K_{c,\vartheta}(\mathrm ds)$, suppose that $\int_\Theta m(c,\vartheta)\{P_c(\mathrm d\vartheta)+Q_c(\mathrm d\vartheta)\}<\infty$ for every $c\in C$.
Define the true and generated posterior-predictive kernels by
\begin{equation*}
\Pi_c^P(B)=\int_\Theta K_{c,\vartheta}(B)\,P_c(\mathrm d\vartheta),\qquad \Pi_c^Q(B)=\int_\Theta K_{c,\vartheta}(B)\,Q_c(\mathrm d\vartheta), \qquad B\in\mathcal B(\R_+).
\end{equation*}
Let $a:C\times\Theta\times\Theta\to[0,\infty]$ be $\mathcal C\otimes\mathcal B(\Theta)\otimes\mathcal B(\Theta)$-measurable, set $a_c(\vartheta,\vartheta')=a(c,\vartheta,\vartheta')$, and define the extended transport cost
\begin{equation*}
W_{a_c}(P_c,Q_c)=\inf_{\gamma\in\Pi(P_c,Q_c)}\int_{\Theta\times\Theta}a_c(\vartheta,\vartheta')\,\gamma(\mathrm d\vartheta,\mathrm d\vartheta').
\end{equation*}

\begin{theorem}\label{thm:integrated-predictive-stability}
Assume that, for $\rho$-almost every $c$, $W_1(K_{c,\vartheta},K_{c,\vartheta'})\leq a_c(\vartheta,\vartheta')$ for all $\vartheta,\vartheta'\in\Theta$,
and $a_c(\vartheta,\vartheta')=L(c)\norm{\vartheta-\vartheta'}$ for a measurable $L:C\to[0,\infty)$. 
Then, for $\rho$-almost every $c$,
\begin{equation}\label{eq:pointwise-predictive-wasserstein-bound}
W_1(\Pi_c^P,\Pi_c^Q)\leq L(c)W_1(P_c,Q_c).
\end{equation}
\end{theorem}

\begin{example}[Poisson frequency]\label{ex:poisson-frequency}
Let $\Theta=(0,\infty)$, let $\tau:C\to(0,\infty)$ be measurable, and define $K_{c,\lambda}=\operatorname{Poisson}(\tau(c)\lambda)$.
Fix $c\in C$ and suppose first that $\lambda'\geq\lambda$.
Let $N\sim\operatorname{Poisson}(\tau(c)\lambda)$ and $M\sim\operatorname{Poisson}(\tau(c)(\lambda'-\lambda))$ be independent, and set $N'=N+M$.
Then $N'\sim\operatorname{Poisson}(\tau(c)\lambda')$, so this coupling yields (recall \cref{eq:w1dist})
\begin{equation*}
W_1(K_{c,\lambda},K_{c,\lambda'})\leq\E[|N'-N|]=\E[M]=\tau(c)(\lambda'-\lambda).
\end{equation*}
Since the identity map on $\R_+$ is $1$-Lipschitz, the dual representation gives
\begin{equation*}
W_1(K_{c,\lambda},K_{c,\lambda'})\geq\left|\E[N']-\E[N]\right|=\tau(c)(\lambda'-\lambda).
\end{equation*}
Exchanging $\lambda$ and $\lambda'$ covers the reverse ordering. Hence $W_1(K_{c,\lambda},K_{c,\lambda'})=\tau(c)|\lambda-\lambda'|$, and \cref{thm:integrated-predictive-stability} applies with $a_c(\lambda,\lambda')=\tau(c)|\lambda-\lambda'|$.
\end{example}

\begin{example}[Pareto severity]\label{ex:pareto-severity}
Let $\Theta\subseteq(1,\infty)$, let $u:C\to(0,\infty)$ be measurable, and let $K_{c,\alpha}$ be the Pareto law with threshold $u(c)$ and shape parameter $\alpha$, whose quantile function is $F_{c,\alpha}^{-1}(p)=u(c)(1-p)^{-1/\alpha}$.
In one dimension, the Wasserstein distance may be expressed in terms of quantile functions.
Since $F_{c,\alpha}^{-1}(p)$ is decreasing in $\alpha$ for every $p\in(0,1)$, direct calculation shows that, for $\alpha,\alpha'\in\Theta$,
\begin{align*}
W_1(K_{c,\alpha},K_{c,\alpha'})
&=\int_0^1\left|F_{c,\alpha}^{-1}(p)-F_{c,\alpha'}^{-1}(p)\right|\,\mathrm dp 
=\left|\int_0^1 F_{c,\alpha}^{-1}(p)\,\mathrm dp-\int_0^1 F_{c,\alpha'}^{-1}(p)\,\mathrm dp\right| \\
&=u(c)\left|\frac{\alpha}{\alpha-1}-\frac{\alpha'}{\alpha'-1}\right|
=u(c)\left|\frac{1}{\alpha-1}-\frac{1}{\alpha'-1}\right|.
\end{align*}
Thus, the bound assumed in \cref{thm:integrated-predictive-stability} holds with
\begin{equation*}
a_c(\alpha,\alpha')=u(c)\left|(\alpha-1)^{-1}-(\alpha'-1)^{-1}\right|.
\end{equation*}
If $\Theta\subseteq[1+\varepsilon,\infty)$ for some $\varepsilon>0$, then $a_c(\alpha,\alpha')\leq u(c)\varepsilon^{-2}|\alpha-\alpha'|$, and \cref{thm:integrated-predictive-stability} applies with $L(c)=u(c)\varepsilon^{-2}$.
\end{example}

\begin{proof}[Proof of \cref{thm:integrated-predictive-stability}]
First, we note that $c\mapsto\Pi_c^P$ and $c\mapsto\Pi_c^Q$ are probability kernels.
Their first moments are, respectively,
\begin{equation*}
\int_{\R_+}s\,\Pi_c^P(\mathrm ds)=\int_\Theta m(c,\vartheta)\,P_c(\mathrm d\vartheta),\qquad\int_{\R_+}s\,\Pi_c^Q(\mathrm ds)=\int_\Theta m(c,\vartheta)\,Q_c(\mathrm d\vartheta).
\end{equation*}
Both quantities are finite by assumption, and hence $\Pi_c^P,\Pi_c^Q\in\mathcal P_1(\R_+)$ for every $c$.
Let $c\in C$ be fixed and such that $W_1(K_{c,\vartheta},K_{c,\vartheta'})\leq a_c(\vartheta,\vartheta')$ for all $\vartheta,\vartheta'\in\Theta$.
Let $\varphi\in\operatorname{Lip}_1(\R_+)$.
The function
\begin{equation*}
h_c(\vartheta)=\int_{\R_+}\varphi(s)\,K_{c,\vartheta}(\mathrm ds)
\end{equation*}
is measurable and integrable under both $P_c$ and $Q_c$, because $|h_c(\vartheta)|\leq m(c,\vartheta) + |\varphi(0)|$.
Kantorovich--Rubinstein duality gives
\begin{equation*}
|h_c(\vartheta)-h_c(\vartheta')|\leq W_1(K_{c,\vartheta},K_{c,\vartheta'})\leq a_c(\vartheta,\vartheta').
\end{equation*}
Thus, for every $\gamma\in\Pi(P_c,Q_c)$,
\begin{equation*}
\left|\int_{\R_+}\varphi\,\mathrm d\Pi_c^P-\int_{\R_+}\varphi\,\mathrm d\Pi_c^Q\right|=\left|\int_{\Theta\times\Theta}\{h_c(\vartheta)-h_c(\vartheta')\}\,\gamma(\mathrm d\vartheta,\mathrm d\vartheta')\right|\leq\int_{\Theta\times\Theta}a_c(\vartheta,\vartheta')\,\gamma(\mathrm d\vartheta,\mathrm d\vartheta').
\end{equation*}
Taking the infimum over $\gamma\in\Pi(P_c,Q_c)$ and then the supremum over all $\varphi\in\operatorname{Lip}_1(\R_+)$ gives
\begin{equation*}
W_1(\Pi_c^P,\Pi_c^Q)\leq W_{a_c}(P_c,Q_c)=L(c)W_1(P_c,Q_c),
\end{equation*}
which proves \eqref{eq:pointwise-predictive-wasserstein-bound}.
\end{proof}

\section{Simulation study}\label{sec:simulation_study}
The preceding results lead directly to the numerical question of how accurately the learned conditional law approximates the posterior law.
In this section, we therefore assess the cWGAN approximation for the Poisson frequency and Pareto severity models, using a single shared generator under mixtures of three prior families.
We assess the performance of the cWGAN generator through calibration diagnostics and comparisons with appropriate analytical or numerical posterior references.
The trained generator is then applied to EM-DAT data to obtain rolling one-year posterior-predictive distributions of aggregate losses from extreme natural catastrophes.
We begin with a description of the natural catastrophe dataset.

We use event-level EM-DAT records for five peril classes: drought, earthquake, flood, storm, and wildfire, which were queried from the EM-DAT database on June 24, 2026.
The Emergency Events Database (EM-DAT, \citet{delforge2025emdat}) records the occurrence and impacts of disasters worldwide from 1900 to the present day.
We restrict to events whose recorded start and end years coincide and lie in 2000--2025 and for which damages are reported.
The equality of the start and end years ensures that every retained event can be attributed to a single year.
We use the reported inflation-adjusted total damage and express all damage values in billions of US dollars for the subsequent numerical analysis.

Figure~\ref{fig:catastrophe_model_calibration} and Table~\ref{tab:emdat-frequency-damage} summarize the dataset.
We fit a power law to the upper tail of the empirical survival function, resulting in a threshold $x_{\min}$ of 2.50 billion USD.
Henceforth, we call an event \emph{extreme} if its adjusted total damage exceeds $x_{\min}$, and retain the resulting 370 events for the subsequent analysis.
As summarized in Table~\ref{tab:emdat-frequency-damage}, storms account for approximately $51\%$ of both events and aggregate damage.
Earthquakes are less frequent but more devastating, contributing $11\%$ of events but $21\%$ of damage, whereas the corresponding shares for floods are $25\%$ and $18\%$.
We then recalibrate the prior moments on expanding windows ending in each year from 2009 through 2025, under the Bühlmann model for frequency and the Bühlmann--Straub model for severity.
The bottom panels in Figure~\ref{fig:catastrophe_model_calibration} summarize the calibrated prior mean and coefficient of variation over the years.

\begin{figure}[tbp]
\centering
\includegraphics[width=.9\linewidth]{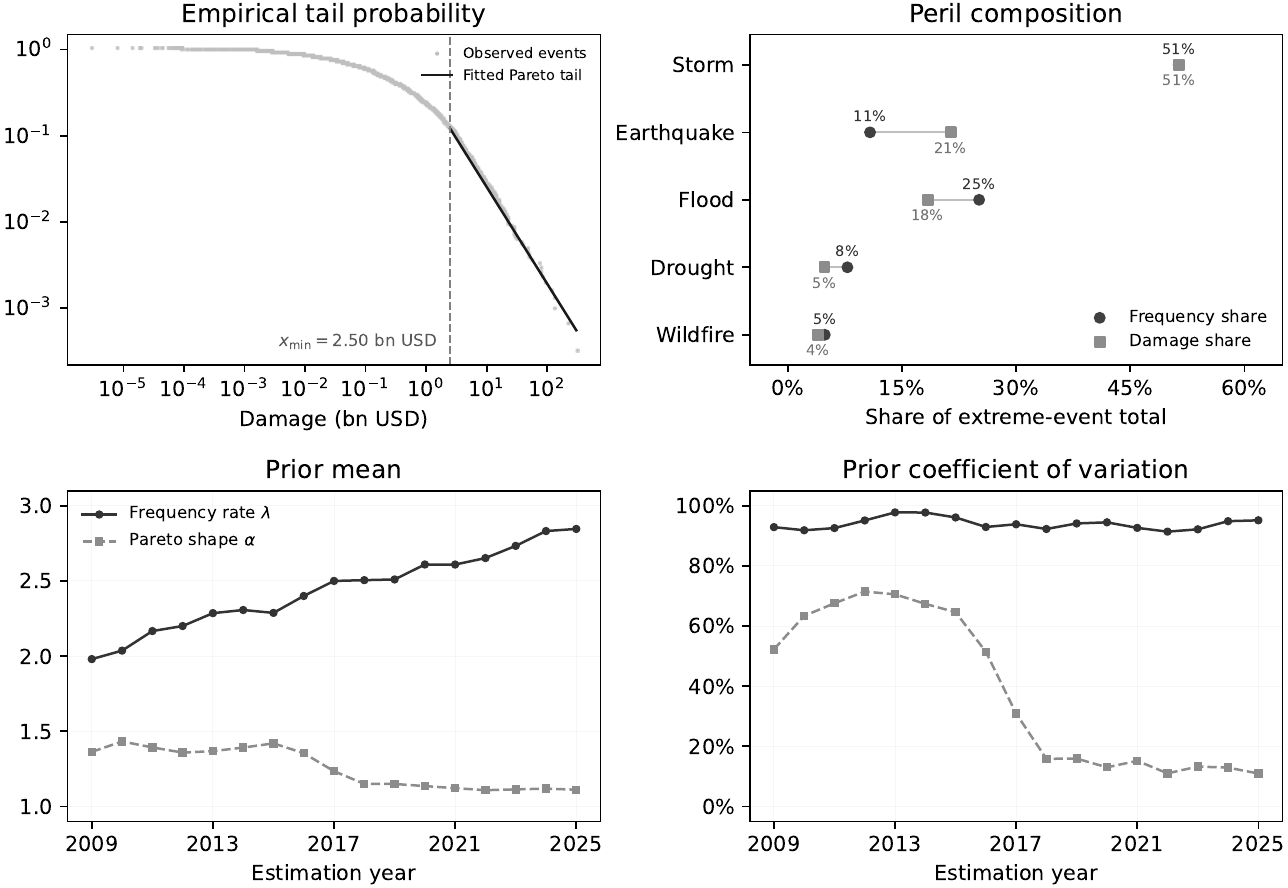}
\caption{Calibration of the catastrophe model.}
\label{fig:catastrophe_model_calibration}
\end{figure}

\begin{table}[tbp]
\centering
\begin{threeparttable}
\begin{tabular}{lrrrrrrrrrr}
\toprule
 & \multicolumn{5}{c}{Frequency} & \multicolumn{5}{c}{Damage (USD billions)} \\
Year & DR & EQ & FL & ST & WF & DR & EQ & FL & ST & WF \\
\midrule
2000 & - & - & 3 & 3 & 1 & - & - & 39.9 & 16.6 & 2.8 \\
2001 & - & 3 & - & 2 & - & - & 11.1 & - & 15.9 & - \\
2002 & 2 & - & 4 & 4 & - & 9.5 & - & 34.9 & 18.3 & - \\
2003 & - & 1 & 4 & 5 & 2 & - & 8.8 & 29.2 & 33.2 & 9.2 \\
2004 & - & 2 & 3 & 9 & - & - & 55.3 & 11.3 & 122.6 & - \\
2005 & - & 1 & 4 & 9 & 2 & - & 8.6 & 16.3 & 282.7 & 6.1 \\
2006 & 1 & 1 & 1 & 4 & - & 4.7 & 5.0 & 5.4 & 15.9 & - \\
2007 & - & 1 & 4 & 5 & 2 & - & 19.4 & 23.9 & 23.8 & 6.6 \\
2008 & - & 1 & 2 & 4 & 1 & - & 127.1 & 18.2 & 64.4 & 3.0 \\
2009 & - & 2 & 1 & 5 & - & - & 7.0 & 3.2 & 17.0 & - \\
2010 & - & 3 & 3 & 6 & 1 & - & 65.7 & 45.1 & 25.4 & 2.7 \\
2011 & 1 & 3 & 5 & 9 & - & 11.4 & 326.3 & 28.3 & 67.3 & - \\
2012 & 1 & 1 & 3 & 8 & - & 28.0 & 22.2 & 18.2 & 101.8 & - \\
2013 & - & 1 & 7 & 9 & - & - & 9.4 & 48.7 & 51.4 & - \\
2014 & 2 & 1 & 3 & 7 & - & 6.4 & 6.8 & 27.3 & 38.7 & - \\
2015 & 2 & 1 & 3 & 4 & - & 9.5 & 7.0 & 9.4 & 15.2 & - \\
2016 & 1 & 4 & 8 & 7 & 1 & 4.0 & 41.5 & 64.7 & 34.2 & 5.4 \\
2017 & 2 & 2 & 2 & 13 & 2 & 6.3 & 10.9 & 11.9 & 327.3 & 20.0 \\
2018 & 2 & 1 & 2 & 6 & 2 & 8.2 & 4.2 & 15.8 & 70.1 & 27.8 \\
2019 & - & - & 6 & 7 & - & - & - & 41.8 & 59.9 & - \\
2020 & 2 & 2 & 6 & 12 & 1 & 9.3 & 16.2 & 51.2 & 80.5 & 13.7 \\
2021 & 3 & 1 & 3 & 5 & 1 & 17.9 & 9.2 & 70.8 & 126.3 & 3.6 \\
2022 & 3 & 2 & 6 & 7 & - & 37.0 & 12.2 & 42.4 & 128.4 & - \\
2023 & 3 & 3 & 2 & 14 & 1 & 20.8 & 52.7 & 13.7 & 92.2 & 5.8 \\
2024 & 2 & 2 & 5 & 17 & - & 11.8 & 18.0 & 38.2 & 180.9 & - \\
2025 & 2 & 1 & 3 & 9 & 1 & 7.9 & 11.0 & 24.4 & 45.8 & 53.0 \\
\midrule
\textbf{Total} & 29 & 40 & 93 & 190 & 18 & 192.7 & 855.5 & 734.4 & 2055.7 & 159.5 \\
\bottomrule
\end{tabular}
\caption{Annual frequency and aggregate damage of extreme natural catastrophes by peril, 2000--2025. Damage is measured in billions of USD. DR denotes drought, EQ earthquake, FL flood, ST storm, and WF wildfire; a dash denotes zero.}
\label{tab:emdat-frequency-damage}
\end{threeparttable}
\end{table}

For both frequency and severity, the likelihood depends on $\vartheta$ through a term of the form $\vartheta^r\exp(-s\vartheta)$, which permits a single generator to be trained for both tasks.
We write the conditioning variable as $c=(m,v,r,s,w)$, where $m$ and $v$ are the prior mean and coefficient of variation and $w=(w_\Gamma,w_{\mathrm{IG}},w_{\mathrm{LN}})$ is the vector of prior-family mixture weights.
For $m>0$, $v>0$, and $w$ in the probability simplex, let $\pi_{w,m,v}$ be the $w$-mixture of Gamma, inverse-Gaussian, and lognormal laws having common mean $m$ and coefficient of variation $v$.
Then $\pi_{w,m,v}$ has mean $m$ and variance $m^2v^2$ for every $w$, so changing $w$ does not change the prior mean or variance.
In either task, the posterior law is
\begin{equation*}
P_c(\mathrm d\vartheta)\propto\vartheta^r\exp(-s\vartheta)\,\pi_{w,m,v}(\mathrm d\vartheta),\qquad (r,s)=
\begin{cases}
(y_+,n), & \text{for frequency},\\
(n,t_n), & \text{for severity}.
\end{cases}
\end{equation*}
We take $\nu_Z=\mathcal N(0,I_2)$.
The generator has two hidden layers of 64 units each with Sigmoid Linear Unit (SiLU) activation and an exponential output transformation to preserve positivity.
The critic has two hidden layers of 64 units each with Leaky Rectified Linear Unit (Leaky ReLU) activation of negative slope $0.2$ and is applied to hybrid $3$-packs corresponding to a common conditioning value.
We use learning rates $3\times10^{-5}$ and $10^{-4}$ for the generator and critic, respectively (cf.\ \citealp{heusel2017gans}), and train the cWGAN for $2\times10^6$ iterations on batches of size 8,192.
The critic is regularized by the one-sided gradient penalty specified in Algorithm~\ref{alg:shared-cwgan-training}, adapted to the hybrid packs, which is a mild version of the penalty meant to enforce the Lipschitz property of the critic \citep{gulrajani2017improved}.
The prior mean $m$ and coefficient of variation $v$ are sampled uniformly from a range that covers all the calibrated values in Figure~\ref{fig:catastrophe_model_calibration}.
The mixture weights are sampled as follows: 30\% are placed at the vertices; 30\% are placed on the edges, where the two nonzero weights are $U$ and $1-U$ with $U\sim\operatorname{Uniform}(0.05,0.95)$; and 40\% are drawn from a $\operatorname{Dirichlet}(1,1,1)$ distribution conditional on every weight being at least $0.05$.
Each minibatch of 8,192 samples contains equal numbers of frequency and severity cases, with $n$ sampled uniformly from $\{1,\ldots,30\}$ and $\{1,\ldots,200\}$, respectively.
Conditional on the sampled $(m,v,w,n)$, we draw $\vartheta\sim\pi_{w,m,v}$ and simulate $(r,s)=(R,n)$ with $R\sim\operatorname{Poisson}(n\vartheta)$ for frequency and $(r,s)=(n,S)$ with $S\sim\operatorname{Gamma}(n,\vartheta)$ for severity.
Thus the conditional law of $\vartheta$ at $c=(m,v,r,s,w)$ is $P_c$.
Algorithm~\ref{alg:shared-cwgan-training} summarizes the training procedure.

\begin{algorithm}[H]
\caption{Training the shared cWGAN with hybrid $k$-packs.}
\label{alg:shared-cwgan-training}
\small
\begin{algorithmic}[1]
\REQUIRE $G_\theta,D_\phi,T,B,n_{\mathrm{crit}},k,\eta_G,\eta_D,
\lambda_\mathrm{GP},u_0,\beta_{\mathrm{EMA}}$.
\STATE Initialize $\theta$ and $\phi$; set $\bar\theta\gets\theta$ and $u\gets0$.
\FOR{$t=1,\ldots,T$}
    \FOR{$\ell=1,\ldots,n_{\mathrm{crit}}$}
        \STATE Draw a batch $\{(c_i,\vartheta_i)\}_{i=1}^{B}$.
        \STATE Draw $Z_{ij}\sim\nu_Z$, $j=1,\ldots, k$, and set
        \begin{equation*}
        \widehat{\boldsymbol\vartheta}_i\gets\bigl(G_\theta(c_i,Z_{i1}),\ldots,G_\theta(c_i,Z_{ik})\bigr).
        \end{equation*}
        \STATE For $j=1,\ldots,k$, let $\boldsymbol\vartheta_i^{(j)}$ be obtained by replacing coordinate $j$ of $\widehat{\boldsymbol\vartheta}_i$ by $\vartheta_i$.
        \STATE Compute
        \begin{equation*}
        \widehat{\mathcal J}\gets\frac{1}{Bk} \sum_{i=1}^{B}\sum_{j=1}^{k} D_\phi(c_i,\boldsymbol\vartheta_i^{(j)})-\frac{1}{B}\sum_{i=1}^{B}D_\phi(c_i,\widehat{\boldsymbol\vartheta}_i).
        \end{equation*}
        \STATE Draw $\varepsilon_i\sim\operatorname{Uniform}(0,1)$ and set
        \begin{equation*}
        \widetilde{\boldsymbol\vartheta}_i^{(j)}\gets \varepsilon_i\boldsymbol\vartheta_i^{(j)}+(1-\varepsilon_i)\widehat{\boldsymbol\vartheta}_i.
        \end{equation*}
        \STATE Compute
        \begin{equation*}
        g_i \gets \frac{1}{k}\sum_{j=1}^{k}\nabla_{\boldsymbol x} D_\phi(c_i,\widetilde{\boldsymbol\vartheta}_i^{(j)}), \qquad \widehat{\mathcal P}\gets\frac{1}{B}\sum_{i=1}^{B}\bigl(\lVert g_i\rVert_2-1\bigr)_+^2.
        \end{equation*}
        \STATE Update $\phi$ by an Adam step of size $\eta_D$ minimizing $-\widehat{\mathcal J}+\lambda_\mathrm{GP}\widehat{\mathcal P}$.
    \ENDFOR
    \STATE Draw a batch, compute $\widehat{\mathcal J}$, and update $\theta$ by an Adam step of size $\eta_G$ minimizing $\widehat{\mathcal J}$ with $\phi$ fixed.
    \STATE Set $u\gets u+1$; if $u=u_0$, set $\bar\theta\gets\theta$; if $u>u_0$, set $\bar\theta\gets \beta_{\mathrm{EMA}}\bar\theta+ (1-\beta_{\mathrm{EMA}})\theta$.
\ENDFOR
\STATE \textbf{Return} $G_{\bar\theta}$.
\end{algorithmic}
\end{algorithm}

We take $n_{\mathrm{crit}}=1$, $\lambda_{\mathrm{GP}}=10$, and Adam coefficients $(\beta_1,\beta_2)=(0,0.9)$.
An exponential moving average (EMA) of the generator parameters is initialized after $u_0=50{,}000$ generator updates and thereafter updated with $\beta_{\mathrm{EMA}}=2^{-1/25{,}000}$, and the averaged generator $G_{\bar\theta}$ is used below.

The cWGAN aims to learn the posterior laws of the frequency and severity parameters under the prior mixture.
Since analytical posterior references are not available throughout the prior-mixture simplex, we first assess $G_{\bar\theta}$ by simulation-based calibration (SBC; \citealp{talts2018validating}).
In simplified notation, suppressing the additional conditioning variables, let $\vartheta\sim\pi$, let $Y\sim p(\,\cdot\mid\vartheta)$, and, conditionally on $Y$, draw $\widetilde{\vartheta}_1,\ldots,\widetilde{\vartheta}_L$ independently from an approximate posterior $Q(\,\cdot\mid Y)$ and independently of $\vartheta$.
Let $f$ be a real-valued measurable test quantity, possibly depending on $Y$, and suppose that the conditional law of $f(Y,\vartheta)$ given $Y$ is atomless.
If $Q(\,\cdot\mid Y)=\mathcal L(\vartheta\mid Y)$, then $\vartheta,\widetilde{\vartheta}_1,\ldots,\widetilde{\vartheta}_L$ are conditionally i.i.d.\ given $Y$, and hence
\begin{equation*}
R=\sum_{\ell=1}^{L}\mathbf{1}_{\{f(Y,\widetilde{\vartheta}_{\ell})<f(Y,\vartheta)\}}
\end{equation*}
is uniformly distributed on $\{0,\ldots,L\}$.

Using 50,000 prior-predictive replications per task, with $L=1023$ approximate posterior draws in each replication, the first column of Figure~\ref{fig:diagnostics} reports the resulting SBC diagnostics.
For each task, we plot the difference between the empirical cumulative distribution function (ECDF) of $\widetilde R=R/L$ and the distribution function of the discrete uniform law on $\{0,1/L,\ldots,1\}$.
For parameter ranks, obtained by taking $f(Y,\vartheta)=\vartheta$, the supremum ECDF deviations are $0.24\%$ for frequency and $0.50\%$ for severity, and both curves remain within the corresponding $95\%$ simultaneous null bands \citep{sailynoja2022graphical}.
The choice of test quantity is nevertheless material.
As shown by \citet{modrak2025simulation}, parameter ranks alone need not detect even the approximation $Q(\,\cdot\mid Y)=\pi$, which ignores the data, whereas data-dependent test quantities such as the joint likelihood can provide sensitivity to such failures.
For the joint log-likelihood, $f(Y,\vartheta)=\log p(Y\mid\vartheta)$, the supremum ECDF deviations are $0.27\%$ and $0.37\%$, respectively, and again both curves remain within the corresponding $95\%$ simultaneous null bands.
For 1,000 draws of $c$ per task, the signed errors in posterior mean, standard deviation, the $95\%$ quantile, and the $95\%$ superquantile relative to deterministic quadrature, normalized by $m$, are centered near zero: all interquartile ranges are contained in $[-1\%,1\%]$, and all central $95\%$ ranges are contained in $[-4\%,4\%]$.
The $W_1$ errors, likewise normalized by $m$, are similarly small: medians are roughly $0.2\%$--$0.5\%$, and the $97.5\%$ quantiles remain below $2.5\%$ across the Gamma, inverse-Gaussian, and lognormal vertices, the edges, and the simplex interior.

\begin{figure}[tbp]
\centering
\includegraphics[width=.9\linewidth]{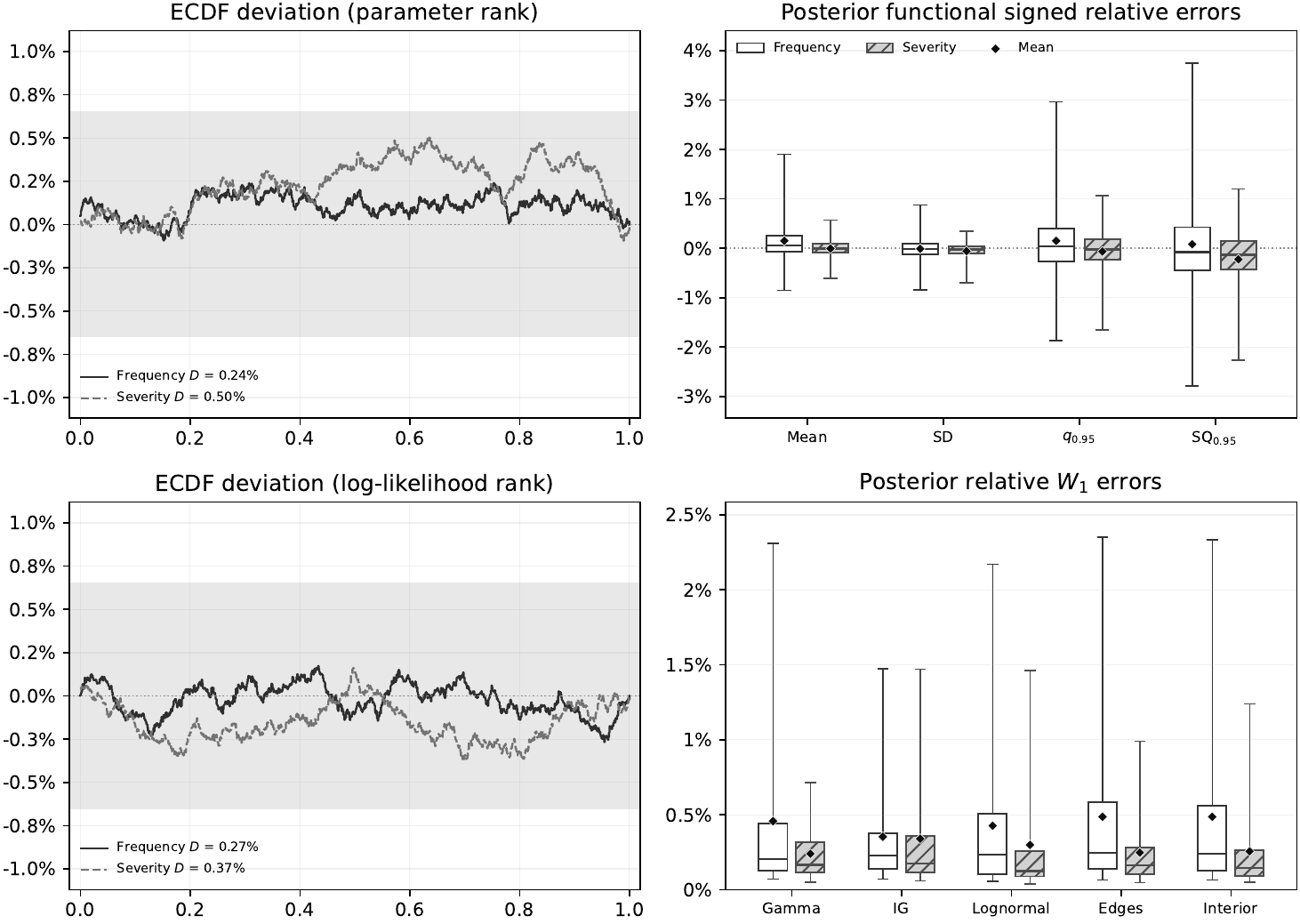}
\caption{Calibration and approximation diagnostics for the shared cWGAN. The figure reports simulation-based calibration for the frequency and severity posterior laws, relative errors in posterior functionals against benchmark references, and relative $1$-Wasserstein errors across the vertices, edges, and interior of the prior-mixture simplex.}
\label{fig:diagnostics}
\end{figure}

In the catastrophe application, the mixture weights $W^F$ and $W^S$ are inferred separately for frequency and severity from independent $\operatorname{Dirichlet}(1,1,1)$ priors.
Each predictive replicate uses one draw of each vector, common to all five perils.
For $x\in\{F,S\}$, write $D^x=(D_1^x,\ldots,D_5^x)$ and assume that, conditional on the corresponding mixture weights $w$, the peril-specific datasets $D_1^x,\ldots,D_5^x$ are independent.
Let $L_j^x(D_j^x\mid\vartheta)$ be the likelihood for peril $j$, with $\vartheta=\lambda$ for $x=F$ and $\vartheta=a$ for $x=S$, and let $\pi_k^x$ denote the common prior law in family $k\in\{\Gamma,\mathrm{IG},\mathrm{LN}\}$.
Then
\begin{equation*}
Z_{jk}^x\coloneqq\int_0^\infty L_j^x(D_j^x\mid\vartheta)\,\pi_k^x(\mathrm{d}\vartheta),
\qquad
p_x(w\mid D^x)\propto\prod_{j=1}^{5}\left(\sum_{k=1}^{3}w_k Z_{jk}^x\right),
\quad
w_k\geq0,\qquad\sum_{k=1}^{3}w_k=1.
\end{equation*}
Thus $Z_{jk}^x$ is the marginal likelihood of $D_j^x$ under $\pi_k^x$, and $\sum_k w_kZ_{jk}^x$ is the corresponding marginal likelihood under mixture weights $w$.
After collecting like terms, the product is a nonnegative linear combination of monomials $w_\Gamma^{n_\Gamma}w_{\mathrm{IG}}^{n_{\mathrm{IG}}}w_{\mathrm{LN}}^{n_{\mathrm{LN}}}$ with $n_\Gamma+n_{\mathrm{IG}}+n_{\mathrm{LN}}=5$.
Hence $p_x(\,\cdot\mid D^x)$ is a mixture of at most $\binom{7}{2}=21$ Dirichlet laws with parameters $(n_\Gamma+1,n_{\mathrm{IG}}+1,n_{\mathrm{LN}}+1)$, from which $W^F$ and $W^S$ are sampled independently.

Figure~\ref{fig:catastrophe_posterior_simplex} reports the rolling one-year posterior-predictive aggregate losses.
Conditional on the threshold $x_{\min}=2.50$ billion USD selected from the 2000--2025 data, each forecast for 2010--2025 uses only information available through the preceding year.
Across these sixteen rolls, the predictive median changes gradually, with an overall upward drift, and every realized aggregate loss lies within the corresponding central $95\%$ predictive interval.
At the 2025 cutoff, the cWGAN forecast for 2026 has a median of $128.2$ billion USD and $95\%$ and $97.5\%$ quantiles of $582.1$ and $985.6$ billion USD, respectively.
A numerical reference, obtained from an extensive MCMC simulation, yields $127.8$, $580.8$, and $986.4$ billion USD, respectively, and the two distributions remain closely aligned over their central mass and through the displayed heavy upper tail.

\begin{figure}[tbp]
\centering
\includegraphics[width=.9\linewidth]{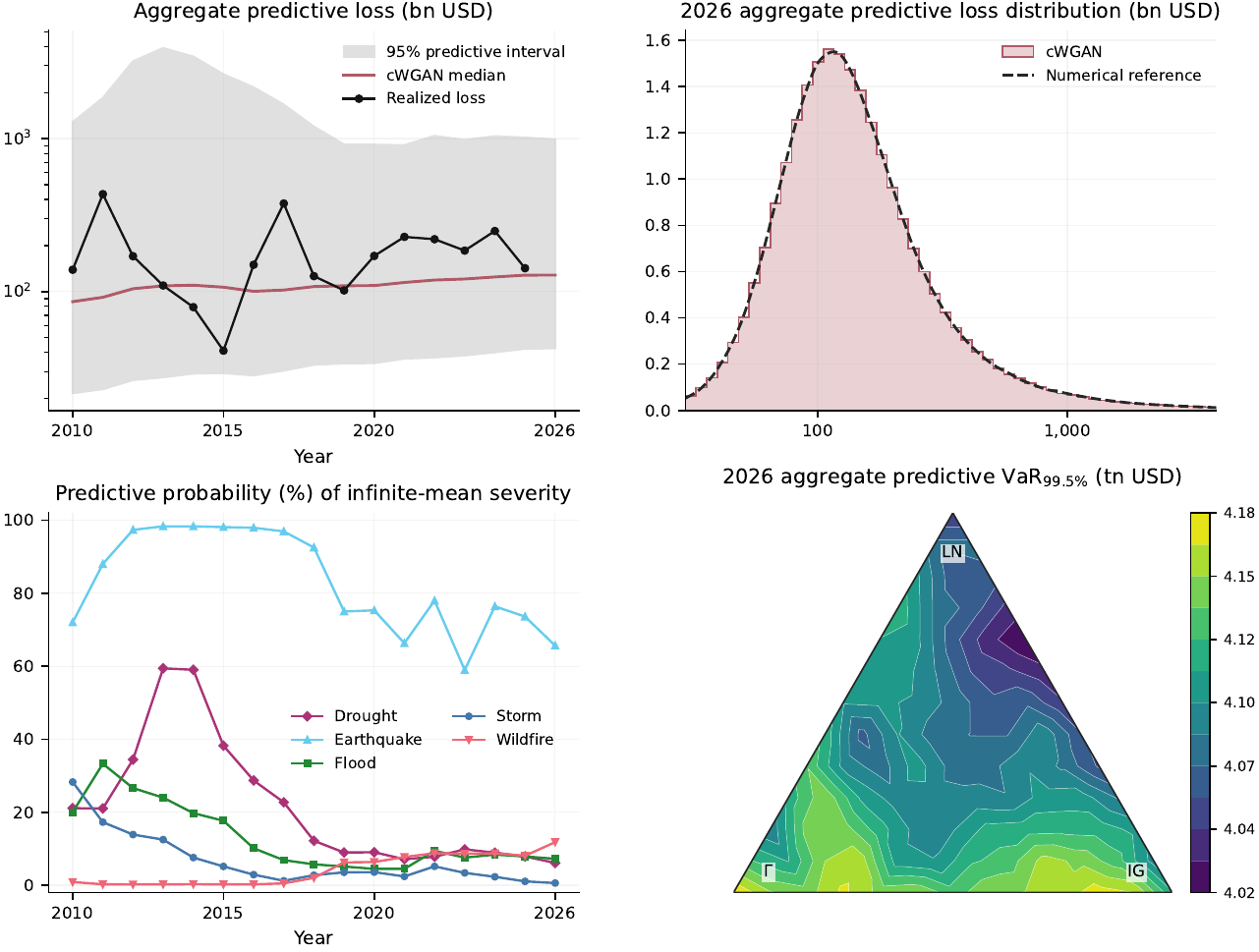}
\caption{Posterior-predictive analysis of extreme natural catastrophe losses. The panels report rolling one-year aggregate-loss forecasts, posterior probabilities that the Pareto shape does not exceed one, the 2026 aggregate-loss predictive distribution with its numerical reference, and the sensitivity of the 2026 aggregate \(\mathrm{VaR}_{99.5\%}\) to the severity-prior mixture weights.}
\label{fig:catastrophe_posterior_simplex}
\end{figure}

The lower-left panel of Figure~\ref{fig:catastrophe_posterior_simplex} explains why the rolling forecasts use medians rather than means.
For each peril $j\in\{1,\ldots,5\}$, let $A_j$ and $\Lambda_j$ denote its Pareto shape parameter and Poisson intensity, respectively.
Positive posterior mass on $A_j\leq1$ makes the posterior-predictive severity mean for peril $j$ infinite and, because $\Lambda_j>0$ almost surely, makes the corresponding peril-level aggregate-loss mean infinite; hence the total aggregate-loss mean is infinite as well.
Positive posterior mass on $A_j\leq2$ likewise precludes a finite second moment for that peril-level aggregate and hence for the total aggregate loss.
Among the five peril-specific shape parameters, the posterior probability of $A_j\leq1$ is largest for earthquakes.
For earthquakes, this probability increases after 2010, coinciding with six extreme events in 2010--2011 totaling $392.0$ billion USD, while the temporary increase in the drought probability follows two events in 2011--2012 totaling approximately $39.4$ billion USD.
The lower-right panel reports the 2026 aggregate 99.5$\%$ Value-at-Risk $\mathrm{VaR}_{99.5\%}$, which ranges from approximately $4.02$ to $4.18$ trillion USD across the severity simplex, with lower values near the inverse-Gaussian--lognormal edge and generally higher values near the Gamma--inverse-Gaussian edge.

\section{Conclusion}\label{sec:conclusion}
In this paper, we formulated repeated Bayesian inference in compound loss models as the amortized approximation of a conditional posterior law.
We developed a conditional Wasserstein GAN that generates draws from the corresponding posterior law conditional on the sample size and a sufficient statistic, the prior mean and coefficient of variation, and the mixture weights of the prior families.
Numerically, a single shared generator approximated the posterior laws of the Poisson intensity and Pareto shape parameter over a broad range of conditioning values.
The approximation was assessed by simulation-based calibration and by comparisons with deterministic quadrature.
In an application to data on extreme natural catastrophe losses, we produced rolling forecasts of aggregate losses from 2010 through 2025, and the resulting 2026 predictive distribution closely matched an extensive MCMC reference.
Moreover, the sensitivity of the aggregate $\operatorname{VaR}_{99.5\%}$ to the severity-prior mixture weights was quantified.

Our results are subject to several qualifications.
The cWGAN is studied numerically for scalar latent parameters, conditionally independent Poisson counts and Pareto severities, a fixed severity threshold, and models admitting low-dimensional sufficient statistics.
The models considered do not include general cross-peril dependence, frequency--severity dependence, covariates, censoring, truncation, reporting delays, or richer multilevel portfolio structures.
Moreover, the reported accuracy of the cWGAN approximation is tied to the conditioning region and prior families used in training and depends on the neural-network architecture and the hyperparameters used for adversarial optimization.
Natural extensions are to construct multivariate posterior samplers for dependent hierarchical loss models and to use learned permutation-invariant summaries when low-dimensional sufficient statistics are unavailable.

\section*{Data availability statement}
The code used for the numerical experiments is available upon reasonable request.

\bibliographystyle{abbrvnat}
\bibliography{cwvi-references}

\end{document}